\documentclass[1p,sort&compress,times]{elsarticle}
\usepackage{amsmath,amssymb,amscd,amsthm}
\usepackage{hyperref}
\usepackage{color}
\makeatletter
\def\ps@pprintTitle{%
 \let\@oddhead\@empty
 \let\@evenhead\@empty
 \def\@oddfoot{}%
 \let\@evenfoot\@oddfoot}
\makeatother
\newcommand{\al}{\alpha}
\newcommand{\be}{\beta}
\newcommand{\de}{\delta}

\newcommand{\vep}{\varepsilon}
\newcommand{\ga}{\gamma}

\newcommand{\la}{\lambda}
\newcommand{\om}{\omega}
\newcommand{\si}{\sigma}

\newcommand{\vp}{\varphi}

\newcommand{\ze}{\zeta}
\newcommand{\Si}{\Sigma}
\newcommand{\Om}{\Omega}
\newcommand{\bh}{\mathbf{h}}

\newcommand{\bn}{\mathbf{n}}

\newcommand{\bs}{\mathbf{s}}

\newcommand{\bx}{\mathbf{x}}
\newcommand{\by}{\mathbf{y}}
\newcommand{\bv}{\mathbf{v}}
\newcommand{\bS}{\mathbf S}

\newcommand{\bsi}{{\boldsymbol{\si}}}
\newcommand{\bxi}{{\boldsymbol{\xi}}}

\newcommand{\hde}{\hat{\de\vphantom{q'}}}
\newcommand{\CC}{{\mathbb C}}

\newcommand{\RR}{{\mathbb R}}
\newcommand{\ZZ}{{\mathbb Z}}
\newcommand{\cE}{{\mathcal E}}

\newcommand{\cH}{{\mathcal H}}

\newcommand{\cN}{{\mathcal N}}
\newcommand{\cP}{{\mathcal P}}

\newcommand{\cZ}{{\mathcal Z}}
\def\Hsc{H^{\mathrm{sc}}}
\def\Zsc{Z^{\mathrm{sc}}}
\newcommand{\tot}{\mathrm{tot}}
\newcommand{\pd}{\partial}

\def\ket#1{|#1\rangle}

\let\ds\displaystyle

\newcommand{\ms}{\mspace{1mu}}

\renewcommand{\le}{\leqslant}
\renewcommand{\ge}{\geqslant}
\def\bbuildrel#1_#2^#3{\mathrel{\mathop{\kern0pt #1}\limits_{#2}^{#3}}}
\newcommand{\tends}[1]{\bbuildrel{\hbox to 2em{\rightarrowfill}}_{#1}^{}}
\newcommand{\erf}{\operatorname{erf}}

\newcommand{\tr}{\operatorname{tr}}

\newcommand{\card}{\operatorname{card}}

\newcommand{\res}{\operatorname{Res}}

\newcommand{\iu}{\mathrm i}
\newcommand{\e}{\mathrm e}

\newcommand{\diff}{\mathrm{d}}

\newcommand{\Iff}{\iff}
\renewcommand{\Im}{\operatorname{Im}}
\newcommand{\sn}{\operatorname{sn}}
\newcommand{\cn}{\operatorname{cn}}

\newcommand{\en}{\enspace}
\newcommand{\all}{\forall}
\newcommand{\pdf}[2]{\frac{\partial #1}{\partial #2}}
\newcommand{\sump}[1]{{\sum_{#1}}'}
\newtheorem{thm}{Theorem}

\newtheorem{lem}{Lemma}

\newcommand{\HS}{\mbox{for the HS chain}}
\newcommand{\PF}{\mbox{for the PF chain}}

\begin{document}
\begin{frontmatter}
  \title{A new perspective on the integrability of Inozemtsev's\\ elliptic spin chain}
\author{Federico Finkel}
\author{Artemio Gonz\'alez-L\'opez\corref{cor}}
\ead{artemio@ucm.es}

\cortext[cor]{Corresponding author}

\address{Departamento de F\'\i sica Te\'orica II, Universidad Complutense de Madrid, 28040 Madrid,
  Spain}
%
%
\begin{abstract}
  The aim of this paper is studying from an alternative point of view the integrability of the
  spin chain with long-range elliptic interactions introduced by Inozemtsev. Our analysis
  relies on some well-established conjectures characterizing the chaotic vs.~integrable
  behavior of a quantum system, formulated in terms of statistical properties of its spectrum.
  More precisely, we study the distribution of consecutive levels of the (unfolded) spectrum, the
  power spectrum of the spectral fluctuations, the average degeneracy, and the equivalence to a
  classical vertex model. Our results are consistent with the general consensus that this model is
  integrable, and that it is closer in this respect to the Heisenberg chain than to its
  trigonometric limit (the Haldane--Shastry chain). On the other hand, we present some numerical
  and analytical evidence showing that the level density of Inozemtsev's chain is asymptotically
  Gaussian as the number of spins tends to infinity, as is the case with the Haldane--Shastry
  chain. We are also able to compute analytically the mean and the standard deviation of the
  spectrum, showing that their asymptotic behavior coincides with that of the Haldane--Shastry
  chain.
\end{abstract}
\begin{keyword}
  Spin chains with long-range interactions\sep Integrability vs.~quantum chaos
\end{keyword}
\end{frontmatter}
\numberwithin{equation}{section}
\section{Introduction}
The celebrated spin chain with long-range interactions introduced independently by
Haldane~\cite{Ha88} and Shastry~\cite{Sh88} in 1988 as a toy model for the one-dimensional Hubbard
Hamiltonian has attracted considerable interest in Mathematical Physics, due to its remarkable
integrability and solvability properties. To name only a few, the $\mathrm{su}(m)$
Haldane--Shastry (HS) chain is completely integrable~\cite{FM93}, it is invariant under the
Yangian $Y\big(\mathrm{su}(m)\big)$ \cite{HHTBP92}, its partition function can be exactly
evaluated for an arbitrary number of spins~\cite{FG05}, and it is equivalent to a simple classical
vertex model~\cite{BBH10}. As a matter of fact, the latter properties can be ultimately traced
back to the connection of this chain to the {\em dynamical} spin Sutherland (trigonometric)
model~\cite{HH92} via the so-called ``freezing trick'' of Polychronakos~\cite{Po93}. Roughly
speaking, in the strong coupling limit the dynamical degrees of freedom of the latter model
decouple from the spin ones, which are in turn governed by the Hamiltonian of the HS chain. Thus,
for instance, the chain's spectrum can be obtained essentially by ``modding out'' the energies of
the scalar Sutherland model~\cite{Su71,Su72} from the spectrum of its spin counterpart.

The fact that the HS chain can be obtained in a suitable limit from a dynamical spin model is
perhaps the feature that sets it apart from other integrable models like, e.g., the Heisenberg
chain. This idea has also been successfully applied to the Calogero (rational)~\cite{Ca71,MP93}
and Inozemtsev (hyperbolic)~\cite{IM86,In96} integrable spin dynamical models, which respectively
yield the Polychronakos--Frahm~\cite{Po93,Fr93} and Frahm--Inozemtsev~\cite{FI94} chains. The
Haldane--Shastry, Polychronakos--Frahm and Frahm--Inozemtsev chains (which shall be collectively
referred to as spin chains of HS type), and their corresponding dynamical models, are all
associated with the $A_{N-1}$ root system~\cite{OP83}, where $N$ is the number of spins. As is
well known, there are generalizations of these models and of their related spin chains to all the
other (extended) root systems (see, e.g.,~\cite{OP83}), although we shall not deal with these more
general models in this paper.

The rational, trigonometric and hyperbolic models of Calogero--Sutherland (CS) type associated
with the $A_{N-1}$ root system are all limiting cases of a more general model with a two-body
interaction potential proportional to a Weierstrass $\wp$ function with suitable
periods~\cite{OP83}. This elliptic CS model, which can be regarded as the $N$-dimensional version
of the Lam\'e potential, is also completely integrable~\cite{OP83}. On the other hand, in spite of
extensive work by many authors (see, e.g.,~\cite{Di93,FV97,Ta00,KT02,La10}), there are no general
explicit formulas for the eigenvalues and eigenfunctions of this model. Even less is known about
its spin version, for which only the lowest integrals of motion~\cite{DI09} and a few particular
solutions (in the three-particle case) have been found~\cite{BI08}. Notwithstanding this, it is
natural to consider the spin chain obtained by applying Polychronakos's freezing trick to the
elliptic CS model. Remarkably, as we shall prove in this paper, this is essentially the chain
introduced by Inozemtsev in 1990 as a model encompassing both the Haldane--Shastry and Heisenberg
chains~\cite{In90}. It should be pointed out that in Inozemtsev's original construction, the
functional form of the spin-spin interaction is obtained not through the freezing trick (which was
actually formulated a few years later), but by requiring that the Hamiltonian possess a quantum
Lax pair analogous to that of the \emph{classical} elliptic CS model.

The interest in Inozemtsev's elliptic chain has been recently rekindled due to its relevance in
connection with the AdS-CFT correspondence. Indeed, this long-range spin chain has been advanced
as a candidate for describing planar $\cN=4$ gauge theory
non-perturbatively~\cite{SS04,BBL09,Re12,Se13}. Since the latter theory is believed to be
integrable in view of the AdS-CFT correspondence, an essential requirement for the validity of
this description is the integrability of the Inozemtsev chain. In fact, already in Inozemtsev's
original paper two first integrals in involution were explicitly exhibited. In a subsequent
work~\cite{In96b}, the same author constructed a large family of first integrals including the two
previously known ones. Although it is conjectured that this family contains a maximal set of first
integrals in involution~\cite{In03}, to the best of our knowledge this crucial fact has never been
rigorously established. Thus, technically speaking the integrability of Inozemtsev's chain still
remains an open question\footnote{It is true that, as remarked above, Inozemtsev's chain admits a
  quantum Lax pair by construction. It should be noted, however, that a \emph{quantum} Lax pair by
  no means guarantees the existence of first integrals; see, e.g.,~\cite{In90}. Indeed, in the
  quantum case the entries of the Lax matrices are \emph{operators}, which in general do not
  commute with each other.}, which has defied the efforts of many specialists in the field for
over two decades.

The main aim of this paper is the analysis of the integrability of Inozemtsev's chain from a
different perspective, based on the statistical properties of its spectrum in the framework of
several important conjectures in quantum chaos. The first of these conjectures, formulated by
Berry and Tabor in the late seventies~\cite{BT77}, posits that the distribution of the spacings
between consecutive (suitably normalized) levels of a ``generic'' integrable quantum system%
\footnote{Well-known exceptions to the Berry--Tabor conjecture are, for instance, the harmonic
  oscillator and the square billiard.} should
be Poissonian. By contrast, in the case of a fully chaotic quantum system the spacings
distribution should follow a Wigner-type law (cf., for instance, Eq.~\eqref{wigner} below). In
fact, the latter conjecture, due to Bohigas, Giannoni and Schmit~\cite{BGS84}, is one of the
cornerstones of the theory of quantum chaos. In view of these conjectures, the behavior of the
spacings distribution of a quantum system can be used as a basic integrability test, relying
exclusively on statistical properties of its spectrum. It is well known, for instance, that the
Heisenberg (spin $1/2$) chain passes this integrability test~\cite{PZBMM93}. In this paper we
will show that this is also the case for Inozemtsev's chain, which provides further evidence of
its integrability.

A related statistical test of integrability, based on treating the energy levels as a time
sequence, has been recently proposed in Ref.~\cite{RGMRF02}. According to this test, the power
spectrum $\cP(\nu)$ of the spectral fluctuations of an integrable quantum system should behave as
$\nu^{-2}$ for sufficiently small values of the frequency $\nu$, while in a chaotic system
$\cP(\nu)\sim\nu^{-1}$. In fact, in Ref.~\cite{FGMMRR04} it is shown that this test is essentially
equivalent to the Berry--Tabor and Bohigas--Giannoni--Schmit conjectures for integrable systems
with Poissonian energy spacings and Gaussian random matrix ensembles. We shall see that the power
spectrum of Inozemtsev's chain behaves as $\nu^{-2}$, hinting again at its integrability. For
comparison purposes, we shall also study the power spectrum of Heisenberg's chain, showing that it
also behaves as the inverse square of the frequency, as expected for an integrable model.

The invariance of the Haldane--Shastry and Polychronakos--Frahm chains under the full Yangian
algebra implies that their spectrum is highly degenerate. Another important consequence of the
Yangian symmetry is the equivalence of the latter chains to a simple classical vertex model
(cf.~Eq.~\eqref{vM} below), whose interactions are closely related to Haldane's
\emph{motifs}~\cite{HHTBP92}. On the other hand, it is known that the Heisenberg chain is only
invariant under the Yangian in the limit when the number of spins tends to infinity~\cite{Be93},
which also explains why its spectrum is much less degenerate. It is widely believed that, like the
Heisenberg chain, Inozemtsev's chain does not possess Yangian symmetry for a finite number of
spins. In this work we shall devise a simple test to rule out the equivalence of a quantum system
to a classical vertex model of the form~\eqref{vM}. Applying this test we shall rigorously show
that the spin $1/2$ Inozemtsev chain is not equivalent to a vertex model of the form~\eqref{vM}
for $10\le N\le 18$ and a wide range of values of the elliptic modulus. This strongly suggests
that this is true in general, which is a further indication of the lack of Yangian symmetry of the
model. Moreover, we shall also see that for $N\gtrsim 10$ the average degeneracy of Inozemtsev's
chain is essentially equal to that of the Heisenberg chain. This is again consistent with the
general consensus that Inozemtsev's chain is an integrable system, more akin to the Heisenberg
than to the original Haldane--Shastry chain (which is obtained by setting the elliptic modulus to
zero).

In spite of the latter conclusion, we shall see that Inozemtsev's elliptic chain shares several
important properties with the HS chain. More precisely, it has been rigorously
established~\cite{EFG10} that the level density of all spin chains of HS type becomes normally
distributed in the large $N$ limit. The parameters of this normal distribution coincide with the
mean and the standard deviation of the spectrum, which have been computed in closed
form~\cite{FG05} and are respectively $O(N^3)$ and $O(N^{5/2})$. By numerically diagonalizing the
Hamiltonian, we shall provide strong evidence that the level density of Inozemtsev's chain is also
asymptotically Gaussian for sufficiently large $N$. To further support this conclusion, we shall
rigorously show that the skewness and the excess kurtosis of the spectrum tend to zero as
$N\to\infty$. We shall also compute in closed form both the mean and the standard deviation of the
spectrum as a function of $N$, $m$ and the elliptic modulus, and show that they have the same
asymptotic behavior when $N\to\infty$ as for the HS chain. To this end, we shall evaluate several
sums involving positive powers of the Weierstrass elliptic function, which are of interest
\emph{per se} as a natural generalization of the corresponding classical formulas for the cosecant
or cotangent~\cite{BY02}.

The paper is organized as follows. In Section 2 we construct the spin chain associated with the
elliptic spin Calogero--Sutherland model by means of Polychronakos's freezing trick, and
explicitly show that it essentially coincides with Inozemtsev's elliptic chain. Section 3 is
devoted to the study of the level density of Inozemtsev's chain, and to the evaluation of the mean
and variance of its spectrum in closed form. We present our statistical analysis of the chain's
integrability in Section 4, where we study the distribution of the spacings between consecutive
levels, the power spectrum of the spectral fluctuations, and the equivalence to a classical vertex
model. In Section 5 we briefly summarize our main results, explaining their relevance and pointing
out some of their applications and connections to other fields. The paper ends with four
appendices covering some background material on elliptic functions, as well as the proofs of some
technical results used in the body of the paper. In particular, in Appendix C we show how to
evaluate in closed form certain finite sums of powers of the Weierstrass elliptic function, and
determine their leading asymptotic behavior.

\section{The spin chain}
Apart from an irrelevant overall factor, the Hamiltonian of the scalar Calogero--Sutherland model
(of $A$ type) with elliptic interactions is commonly taken as
\begin{equation}\label{Hscwp}
\Hsc = -\Delta+a(a-1)\sum_{1\le i\ne j\le N}\wp(x_i-x_j)\,.
\end{equation}
Note that the half-periods $\om_1,\om_3$ of the Weierstrass elliptic function\footnote{We refer
  the reader to~\ref{app.prel} for a brief summary of several basic properties of the Weierstrass
  and Jacobi elliptic functions needed in the sequel.} $\wp$ are assumed to satisfy the condition
$\om_1,\iu\om_3\in\RR$, so that $\wp(x_i-x_j)$ is real and, consequently, $\Hsc$ is self-adjoint
in the Weyl alcove
\[
x_1<\cdots<x_N<x_1+2\om_1\,.
\]
The $\mathrm{su}(m)$ spin version of the previous model is given by
\begin{equation}\label{Hwp}
H = -\Delta+\sum_{1\le i\ne j\le N}\wp(x_i-x_j)\,a(a-S_{ij})\,,
\end{equation}
where $S_{ij}$ is the operator permuting the $i$-th and $j$-th spins. More precisely, the action
of $S_{ij}$ on an element
$\ket{\bs}\equiv\ket{s_1,\dots,s_N}\equiv\ket{s_1}\otimes\cdots\otimes\ket{s_N}$, $1\le s_i\le m$,
of the canonical basis of the internal Hilbert space $\Si\equiv(\CC^m)^{\otimes N}$ is given by
\[
S_{ij}\ket{s_1,\dots,s_i,\dots,s_j,\dots,s_N}=\ket{s_1,\dots,s_j,\dots,s_i,\dots,s_N}\,.
\]
Using Eqs.~\eqref{eis}, \eqref{wpsngen} and \eqref{kei}, and performing the change of variables
$x_i'=\sqrt{e_1-e_3}\,x_i$, we immediately obtain
\[
H=(e_1-e_3)\bigg(-\Delta'+\sum_{1\le i\ne j\le
  N}\frac{a(a-S_{ij})}{\sn^2(x_i'-x_j')}\bigg)+e_3\sum_{1\le i\ne j\le N}a(a-S_{ij})\,.
\]
The constant operator $\sum_{1\le i\ne j\ne N}S_{ij}$ commutes with $x_i$, $\pd_{x_i}$ and
$S_{ij}$ for all $i,j$, and hence with $H$. Therefore the last term in $H$ can be omitted without
changing neither the integrability nor the solvability of the model. We shall thus equivalently
define the Hamiltonian of the elliptic spin Calogero--Sutherland model as
\begin{equation}\label{Hsn}
  H=-\Delta+\sum_{1\le i\ne j\le
    N}\frac{a(a-S_{ij})}{\sn^2(x_i-x_j)}\,.
\end{equation}
This definition has several practical advantages over the usual one. In particular, the operator
$H$ in the latter equation depends on a single real parameter $k\in[0,1)$ (the modulus of the
elliptic function), the value $k=0$ yielding the well-known expression for the spin Sutherland
model. The scalar counterpart of Eq.~\eqref{Hsn} is
\begin{equation}
  \label{Hscsn}
  \Hsc = -\Delta+\sum_{1\le i\ne j\le
    N}\frac{a(a-1)}{\sn^2(x_i-x_j)}\,,
\end{equation}
which differs from the Hamiltonian~\eqref{Hscwp} by a trivial (constant) rescaling and the
addition of a constant energy. Note that the configuration space of both $H$ and $\Hsc$ can be
taken as the Weyl alcove
\begin{equation}
  \label{Walc}
  A=\big\{\,\bx\in\RR^N:x_1<\cdots<x_N<x_1+2K\big\}\,,
\end{equation}
where $2K$ is the real period of the elliptic function $\sn^2$ (cf.~Eq.~\eqref{K}).

From Eqs.~\eqref{Hsn}-\eqref{Hscsn} it immediately follows that
\[
H = \Hsc +4a \hat H(\bx)\,,
\]
with
\[
\hat H(\bx)=\frac12\sum_{1\le i<j\le N}\frac{1-S_{ij}}{\sn^2(x_i-x_j)}\,.
\]
In order to define the spin chain associated with the spin dynamical model~\eqref{Hsn}, we need to
study the equilibria in the Weyl alcove~\eqref{Walc} of the function
\begin{equation}
  \label{U}
  U(\bx)=\sum_{1\le i\ne j\le
    N}\frac1{\sn^2(x_i-x_j)}\,,
\end{equation}
which is proportional to the interaction potential of the scalar model~\eqref{Hscsn}. We shall
show in Subsection~\ref{sec.sites} that (up to a rigid translation) $U$ has a unique minimum
$\bxi\equiv(\xi_1,,\dots,\xi_N)$ in the set $A$, whose coordinates are given by
\begin{equation}\label{xij}
\xi_j=\frac{2jK}N\,,\qquad j=1,\dots,N\,.
\end{equation}
According to Polychronakos's freezing trick argument~\cite{Po93}, the Hamiltonian of the spin
chain associated with $H$ is (proportional to)
\begin{equation}\label{cH}
  \cH=\hat H(\bxi)
\equiv\frac12\sum_{1\le i<j\le N}\frac{1-S_{ij}}{\sn^2\bigl(2(i-j)\frac KN\bigr)}\,.
\end{equation}
Moreover, it can be shown~\cite{Po94} that the partition functions $\cZ$, $Z$ and $\Zsc$ of $\cH$,
$H$ and $\Hsc$ are related by
\[
\cZ(T)=\lim_{a\to\infty}\frac{Z(4aT)}{\Zsc(4aT)}\,.
\]
In order to establish the equivalence between Eq.~\eqref{cH} and the original definition of
Inozemtsev~\cite{In90}, it suffices to note that by Eqs.~\eqref{wpdil} and~\eqref{wpsn} we have
\begin{equation}\label{snwp}
\sn^{-2}\biggl(2(i-j)\frac KN\biggr)=\wp\bigl(2(i-j)\tfrac KN;K,\iu K'\bigr)+\frac13\,(1+k^2)
=\frac{N^2}{4K^2}\,\wp_N(i-j)+\frac13\,(1+k^2)\,,
\end{equation}
where $\wp_N$ is the Weierstrass function with periods $N$ and~$\,\iu NK'/K\equiv\iu N\tau$.

{}It is obvious on account of Eq.~\eqref{snsin} that when $k=0$ the elliptic chain~\eqref{cH}
reduces exactly to the Haldane--Shastry chain
\begin{equation}\label{HSchain}
\cH_{\mathrm{HS}}=\frac12\sum_{1\le i<j\le N}\frac{1-S_{ij}}{\sin^2\bigl((i-j)\frac\pi N\bigr)}\,.
\end{equation}
(This is, indeed, the main reason to modify
Inozemtsev's original definition.) It is also of interest to study the limit of the
Hamiltonian~\eqref{cH} when $k\to 1$. To this ends, we shall make use of the
identity
\begin{equation}
  \label{qseries}
  \frac{\om_1^2}{\pi^2}\,\wp(z;\om_1,\om_3)=-\frac1{12}+\frac14\,\sin^{-2}\biggl(\frac{\pi
    z}{2\om_1}\biggr)+4\sum_{j=1}^\infty\frac{j\,q^{2j}}{1-q^{2j}}\,\sin^2\biggl(\frac{j\,\pi
    z}{2\om_1}\biggr)
\end{equation}
(cf.~\cite{OLBC10}, Eqs.~23.8.1-5), where $\Im(\om_3/\om_1)>0$ and
\[
q=\e^{\iu\pi\om_3/\om_1}
\]
is the so-called Jacobi nome. From Eqs.~\eqref{snwp}-\eqref{qseries} with $\om_1=-\iu K'$,
$\om_3=K$ and $z=2lK/N\equiv 2l\om_3/N$, after a straightforward calculation we obtain
\begin{equation}\label{snqseries}
\sn^{-2}(2lK/N)=\frac13(1+k^2)+\frac{\pi^2}{K'^2}\ms\left(\frac1{12}+\frac{q^{2l/N}}{(1-q^{2l/N})^2}+
\sum_{j=1}^\infty\frac{j\,q^{2j(N-l)/N}(1-q^{2jl/N})^2}{1-q^{2j}}\right)\,.
\end{equation}
Letting $k\to1$ and noting that $K'(1)=\pi/2$ and $q\to0$ as $k\to1$ we finally have
\[
\lim_{k\to1}\sn^{-2}(2lK/N)=1\,,\qquad 1\le l\le N-1\,,
\]
and therefore
\begin{equation}\label{cH1}
\lim_{k\to1}\cH=\frac12\sum_{i<j}(1-S_{ij})\equiv\cH_1\,.
\end{equation}
As first pointed out by Inozemtsev, the Heisenberg chain is also related to the $k\to1$ limit of
the elliptic chain. Indeed, when $k\to1$ the nome $q$ and the imaginary half-period $\iu K'$
satisfy\footnote{We shall say that $f(x)=O(x^n)$ as $x\to0$ or $x\to\infty$ if $|x^{-n}f(x)|$ is
  bounded in this limit.}
\[
q=\frac{\vep}{16}+O(\vep^2)\,,\qquad K'=\frac{\pi}2\bigg(1+2\sum_{n=1}^\infty q^{n^2}\bigg)^2\,,
\]
where $\vep\equiv 1-k^2$ (see~Ref.~\cite{OLBC10}). Hence
\[
\frac{\pi^2}{4K'^{2}}=1-8q+O(q^2)=1-\frac{\vep}2+O(\vep^2)\,,
\]
and therefore
\begin{align*}
\sn^{-2}(2lK/N)&=\frac{2}{3}-\frac{\vep}3+4\bigg(1-\frac{\vep}2+O(\vep^2)\bigg)\,\left(
                 \frac1{12}+q^{2/N}\,\frac{q^{2(l-1)/N}}{(1-q^{2l/N})^2}\right.\\
  &\hphantom{=\frac{2}{3}-\frac{\vep}3+4\bigg(1-\frac{\vep}2+O(\vep^2)\bigg)\,
                 \left(\frac1{12}\right.}\left.{}+
q^{2/N}\sum_{j=1}^\infty\frac{j\,q^{2[j(N-l)-1]/N}(1-q^{2jl/N})^2}{1-q^{2j}}\right)\\
&=1+4\bigg(\frac{\vep}{16}\bigg)^{2/N}\big(\de_{l,1}+\de_{l,N-1}\big)+O(\vep^{4/N})\,,
\end{align*}
where we have dropped an $O(\vep)$ term negligible compared to $\vep^{4/N}$ for $N>4$. Thus
\[
\cH=\cH_1+\bigg(\frac{\vep}{16}\bigg)^{2/N}\,\cH_{\mathrm{He}}+O(\vep^{4/N})\,,
\]
where
\begin{equation}\label{cHhe}
\cH_{\mathrm{He}}=2\sum_{i}(1-S_{i,i+1})\,,\qquad S_{N,N+1}\equiv S_{1N}\,.
\end{equation}
As is well known~\cite{Po06}, the spin permutation operators
$S_{ij}$ can be expressed in terms of the $\mathrm{su}(m)$ Hermitian generators $J_k^\al$ at site
$k$ (with the normalization $\tr(J^\al_kJ^\be_k)=2\de^{\al\be}$) as
\begin{equation}\label{Sij}
S_{ij}=\frac1m+\frac12\sum_{\al=1}^{m^2-1}J^\al_iJ^\al_j\,.
\end{equation}
Thus, for $m=2$ the operator $\cH_{\mathrm{He}}$ is essentially the Hamiltonian of the spin $1/2$
(closed) Heisenberg chain.

\subsection{The chain sites}\label{sec.sites}
We shall now prove that the chain sites of the spin chain associated with the elliptic spin
Calogero--Sutherland model~\eqref{Hsn} are indeed the points~\eqref{xij}, so that Eq.~\eqref{cH}
is essentially equivalent to Inozemtsev's original definition. As we have seen above, it suffices
to show that the scalar potential $U$ has a unique minimum (up to a rigid translation) in the Weyl
alcove~\eqref{Walc}, with coordinates given by Eq.~\eqref{xij}. To this end, we shall first prove
the following more general result:
\begin{thm}\label{thm.critpoint}
  Let $f$ be an odd differentiable function with no zeros in the interval $(0,2K)$, satisfying
  \begin{equation}\label{fcondK}
  f(x+2K)=-f(x)\,,\qquad\all x\,.
\end{equation}
Then the point $\bxi=(\xi_1,\dots,\xi_N)$ with coordinates~\eqref{xij}
is a critical point of the scalar potential
\begin{equation}\label{Uf}
U(\bx)=\sum_{1\le i\ne j\le N} f^{-2}(x_i-x_j)\,.
\end{equation}
\end{thm}
\begin{proof}
  Note, first of all, that from Eq.~\eqref{fcondK} it follows that $f^2$ is $2K$-periodic.
  Moreover, differentiating the latter equation we obtain
  \[
  f'(x+2K)=-f'(x)\,,\qquad\all x\,,
  \]
  and since $f'$ is even we have
  \[
  f'(K)=-f'(-K)=-f'(K)\quad\implies\quad f'(K)=0\,.
  \]
  Let $g=-2f'/f^3$ denote the derivative of $f^{-2}$, which is defined everywhere except at even
  multiples of $K$. By the previous remarks, $g$ is an odd, $2K$-periodic function satisfying
  \begin{equation}\label{gK}
    g(K)=0\,.
  \end{equation}
 We must prove that
  \begin{equation}\label{xiconds}
    \frac12\,\frac{\pd U}{\pd x_i}(\bxi)=\sum_{j;j\ne i}g(\xi_i-\xi_j)
    = \sum_{j;j\ne i}g\bigl(2(i-j)\tfrac KN\bigr)=0\,,\qquad i=1,\dots,N\,.
  \end{equation}
  Calling $l=i-j$ we have
  \begin{align*}
    \sum_{j;j\ne i}g\bigl(2(i-j)\tfrac KN\bigr)&=\sum_{l=1}^{i-1}g\bigl(\tfrac{2lK}N\bigr)+
    \sum_{l=i-N}^{-1}g\bigl(\tfrac{2lK}N\bigr)=\sum_{l=1}^{i-1}g\bigl(\tfrac{2lK}N\bigr)+
    \sum_{l=i-N}^{-1}g\bigl(\tfrac{2(l+N)K}N\bigr)\\&= \sum_{l=1}^{i-1}g\bigl(\tfrac{2lK}N\bigr)+
    \sum_{l=i}^{N-1}g\bigl(\tfrac{2lK}N\bigr)
    =\sum_{l=1}^{N-1}g\bigl(\tfrac{2lK}N\bigr)\,,
  \end{align*}
  where in the second equality we have used the $2K$-periodicity of $g$. Thus the $N$
  equations~\eqref{xiconds} are equivalent to the single condition
  \begin{equation}\label{gsum}
    \sum_{l=1}^{N-1}g\bigl(\tfrac{2lK}N\bigr)=0\,.
  \end{equation}
  In order to establish the previous identity, note that for arbitrary $g$ we have
  \begin{equation}\label{gsum2}
    \sum_{l=1}^{N-1}g\bigl(\tfrac{2lK}N\bigr)=
    \sum_{l=1}^{[(N-1)/2]}\left[g\bigl(\tfrac{2lK}N\bigr)+g\bigl(\tfrac{2(N-l)K}N\bigr)\right]
    +\big(1-\pi(N)\bigr)g(K)\,,
  \end{equation}
  where $\pi(N)$ is the parity of $N$. However, in this case the last term can be omitted on
  account of Eq.~\eqref{gK}, while
  \[
  g\bigl(\tfrac{2(N-l)K}N\bigr)=g\bigl(2K-\tfrac{2lK}N\bigr)=g\bigl(-\tfrac{2lK}N\bigr)
  =-g\bigl(\tfrac{2lK}N\bigr)\,,
  \]
  so that the sum in the right-hand side of Eq.~\eqref{gsum2} also vanishes. This proves
  Eq.~\eqref{gsum} and hence establishes our claim.
\end{proof}

The previous theorem with $f(x)=\sn x$ implies that the point $\bxi$ with coordinates~\eqref{xij}
is a critical point of the potential~$U$, which obviously lies in the open region~\eqref{Walc}. It
is not clear, however, whether this is the only critical point of $U$ in the latter region, and
whether it is indeed a minimum. It turns out that both of these facts can be established in a
rather straightforward way for a large class of potentials of the form~\eqref{Uf}. We must first
deal with a technical complication stemming from the fact that $U$ is invariant under a rigid
translation of the particles, i.e,
\[
U(\bx+\la\bv)=U(\bx)\,,\qquad \all\la\in\RR\,,\quad \bv=(1,\dots,1)\,.
\]
This problem is addressed in~\ref{app.lemmas}, where several standard criteria are generalized to
this type of functions. With the help of these criteria, it is straightforward to obtain the
following general result:
\begin{thm}
  Let $U(\bx)$ be defined by equation~\eqref{Uf}, where $f$ is as in
  Theorem~\ref{thm.critpoint}, and assume furthermore that $f$ is of class $C^2$ and  \begin{equation}\label{fcond}
    3f'^2(x)-f(x)f''(x)>0\,,\qquad \all x\in(0,2K)\,.
  \end{equation}
  Then the point $\bxi$ with coordinates~\eqref{xij} is the unique critical point of $U$ in the
  set $A$ modulo rigid translations of the particles, and is in fact the global minimum of $U$ in
  the latter set.
\end{thm}
\begin{proof}
  Note, first of all, that the set $A$ in Eq.~\eqref{Walc} is obviously convex and invariant under
  an overall translation of the particles' coordinates. Moreover, since $0<|x_i-x_j|<2K$ for all
  $\bx\in A$, the function $U$ is of class $C^2$ in $A$. In order to apply Lemma~\ref{lem.minU},
  we must evaluate the Hessian of $U$ at an arbitrary point of $A$. Calling again $g=(f^{-2})'$ we
  easily obtain
  \[
  \pdf{^2U}{x_i^2}=2\sum_{j;j\ne i}g'(x_i-x_j)\,,\qquad \pdf{^2U}{x_i\partial
    x_j}=-2g'(x_i-x_j)\quad (i\ne j)\,,
  \]
  where we have taken into account that $f^2$ is an even function. If $\bh\in\RR^N$ is an arbitrary
  vector we therefore have
  \[
  \Big(\bh,D^2U(\bx)\cdot\bh\Big)=2\sum_{i\ne j}g'(x_i-x_j)(h_i^2-h_ih_j)\,.
  \]
  However, the function $g'\equiv (f^{-2})''$ is even (cf.~Theorem~\ref{thm.critpoint}), so
  that
  \[
  \sum_{i\ne j}g'(x_i-x_j)h_i^2=\frac12\sum_{i\ne j}g'(x_i-x_j)(h_i^2+h_j^2)\,,
  \]
  and therefore
  \[
  \Big(\bh,D^2U(\bx)\cdot\bh\Big)=\sum_{i\ne j}g'(x_i-x_j)(h_i-h_j)^2\,.
  \]
  Since
  \[
  g'=(f^{-2})''=\frac2{f^4}\,(3f'^2-ff'')
  \]
  is $2K$-periodic, by Eq.~\eqref{fcond} it is positive everywhere except at even multiples of
  $K$. Therefore the Hessian $D^2U(\bx)$ is positive semidefinite, and moreover
  \[
  \Big(\bh,D^2U(\bx)\cdot\bh\Big)=0\en\Iff\en h_i=h_j\,,\en\all i\ne j\en\Iff\en\bh\in\RR\bv\,,
  \]
  with $\bv=(1,\dots,1)$. The statement of the theorem then follows directly from
  Lemma~\ref{lem.minU}.
\end{proof}
In the particular case of the potential $U$ in Eq.~\eqref{U}, for which $f(x)=\sn x$ by
Eq.~\eqref{Uf}, we have
  \[
  3f'^2(x)-f(x) f''(x)=k^2\cn^4x+2\cn^2x+k'^2>0\,,\qquad\all x\in\RR\,.
  \]
  By the previous theorem, $U$ has a unique critical point $\bxi\in A$ (modulo a rigid
  translation), whose coordinates are given by Eq.~\eqref{xij}. Moreover, this point is the
  global minimum of $U$ in the open set $A$, as claimed.
  
\section{Level density}\label{sec.LD}
A well-known property shared by all spin chains of HS type is the fact that their level density
becomes asymptotically Gaussian when the number of spins is large enough. This property has been
rigorously established for all HS chains of type $A_{N-1}$~\cite{EFG10}, and there is strong
numerical evidence that it also holds for other root systems and in the supersymmetric
case~\cite{EFGR05,BB06,BFGR08,BFGR09,BFG09,BB09,BFG13}. Since the Haldane--Shastry chain is the
$k\to0$ limit of Inozemtsev's chain, it is natural to investigate whether the level density of the
latter chain is also approximately Gaussian when $N\gg1$. A fundamental difficulty that must be
faced in this case is the lack of a closed-form expression for the spectrum, so that one is forced
to diagonalize numerically the Hamiltonian $\cH$ in Eq.~\eqref{cH}. If the symmetries of~$\cH$ are
not taken into account, this approach becomes unfeasible on a standard desktop computer for
$N\gtrsim 15$ spins in the most favorable case $m=2$. Fortunately, the fact that $\cH$ depends
only on spin permutation operators implies that it leaves invariant all the subspaces with
well-defined ``spin content''. More precisely, we shall say that two basic states $\ket\bs$,
$\ket{\bs'}$ have the same spin content if their spin components are related by a suitable
permutation. {}From this definition, it immediately follows that each linear subspace of $\Si$
spanned by {\em all} basic states with the same spin content is invariant under the permutation
operators $S_{ij}$, and hence under~$\cH$.

Diagonalizing the restriction of $\cH$ to each of these invariant subspaces we have been able to
compute the whole spectrum for up to $N=18$ spins ($m=2$) and $N=12$ ($m=3$), for several values
of the modulus $k$. Our numerical results clearly indicate that for sufficiently large $N$ the
level density of Inozemtsev's elliptic chain follows with great accuracy a Gaussian distribution
with parameters $\mu$ and $\si$ equal to the mean and standard deviation of the spectrum. For
instance, in Fig.~\ref{fig.qqplots} we have plotted the quantiles of the Gaussian distribution
with parameters $\mu$ and $\si$ computed from the spectrum versus the corresponding quantiles of
the level density of Inozemtsev's chain in the cases $N=18$, $m=2$ and $N=12$, $m=3$. It is
apparent that in both cases the curve thus obtained practically coincides with the straight line
$y=x$ for a wide range of energies. In fact, these plots are very similar to the corresponding one
for the Haldane--Shastry chain with $N=18$, $m=2$, whose level density has been rigorously shown
to be asymptotically Gaussian when the number of sites tends to infinity~\cite{EFG10}. In
contrast, the analogous $q$-$q$ plot for the Heisenberg chain with $N=18$ spins, also shown in
Fig.~\ref{fig.qqplots}, deviates more markedly from the line $y=x$. The above remarks are
quantitatively corroborated by the numerical data summarized in Table~\ref{table.RMSE}, where we
have listed the RMS error of the fit of the spectrum's cumulative level density to the cumulative
Gaussian distribution
\begin{equation}\label{CGDF}
G(E)=\frac12\,\bigg[1+\erf\bigg(\frac{E-\mu}{\sqrt 2 \si}\bigg)\bigg]\,,
\end{equation}
where $\mu$ and $\si$ denote respectively the mean and the standard deviation of the spectrum.
\begin{figure}[h]
  \centering
  \includegraphics[width=6.5cm]{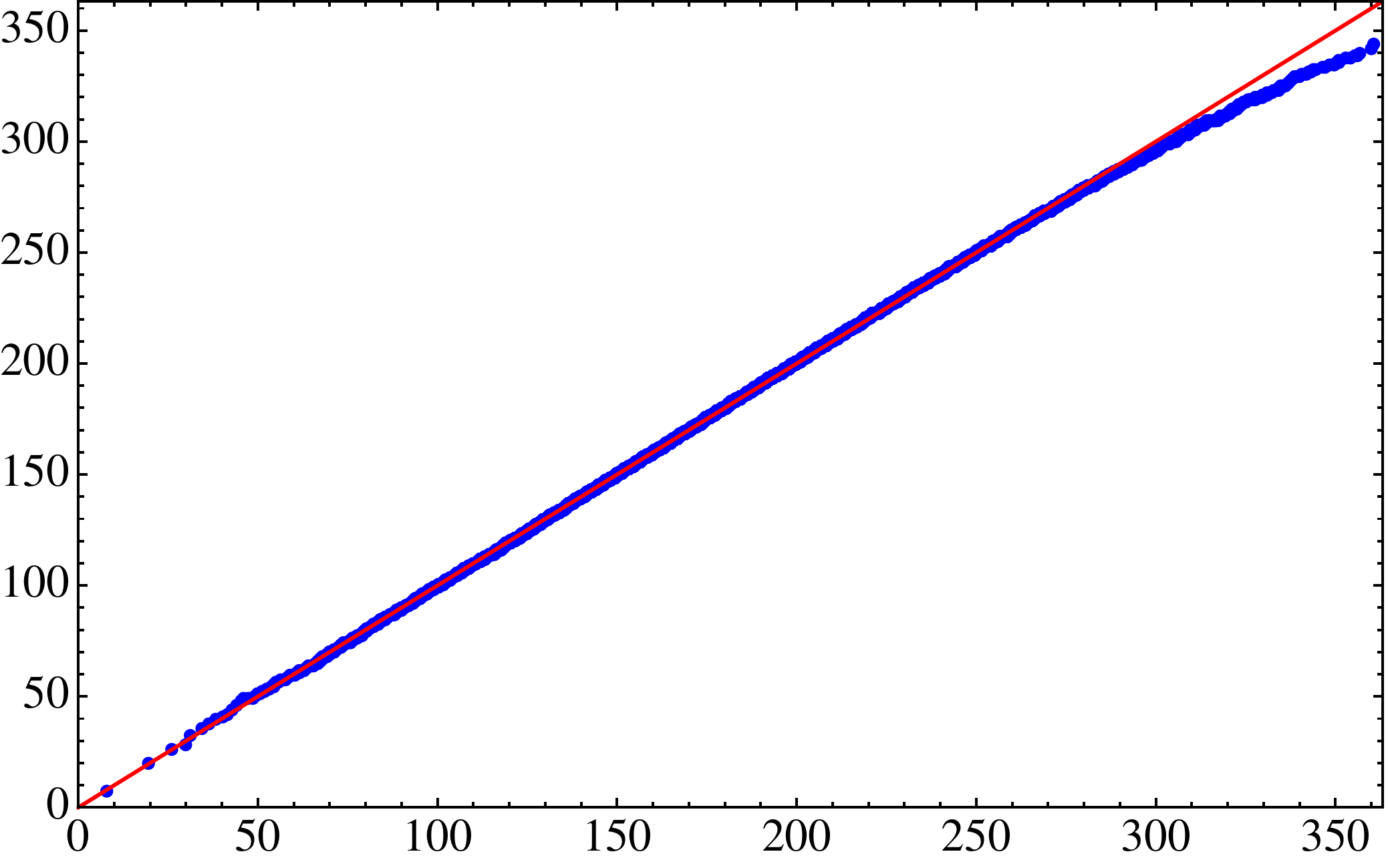}\hfill
  \includegraphics[width=6.5cm]{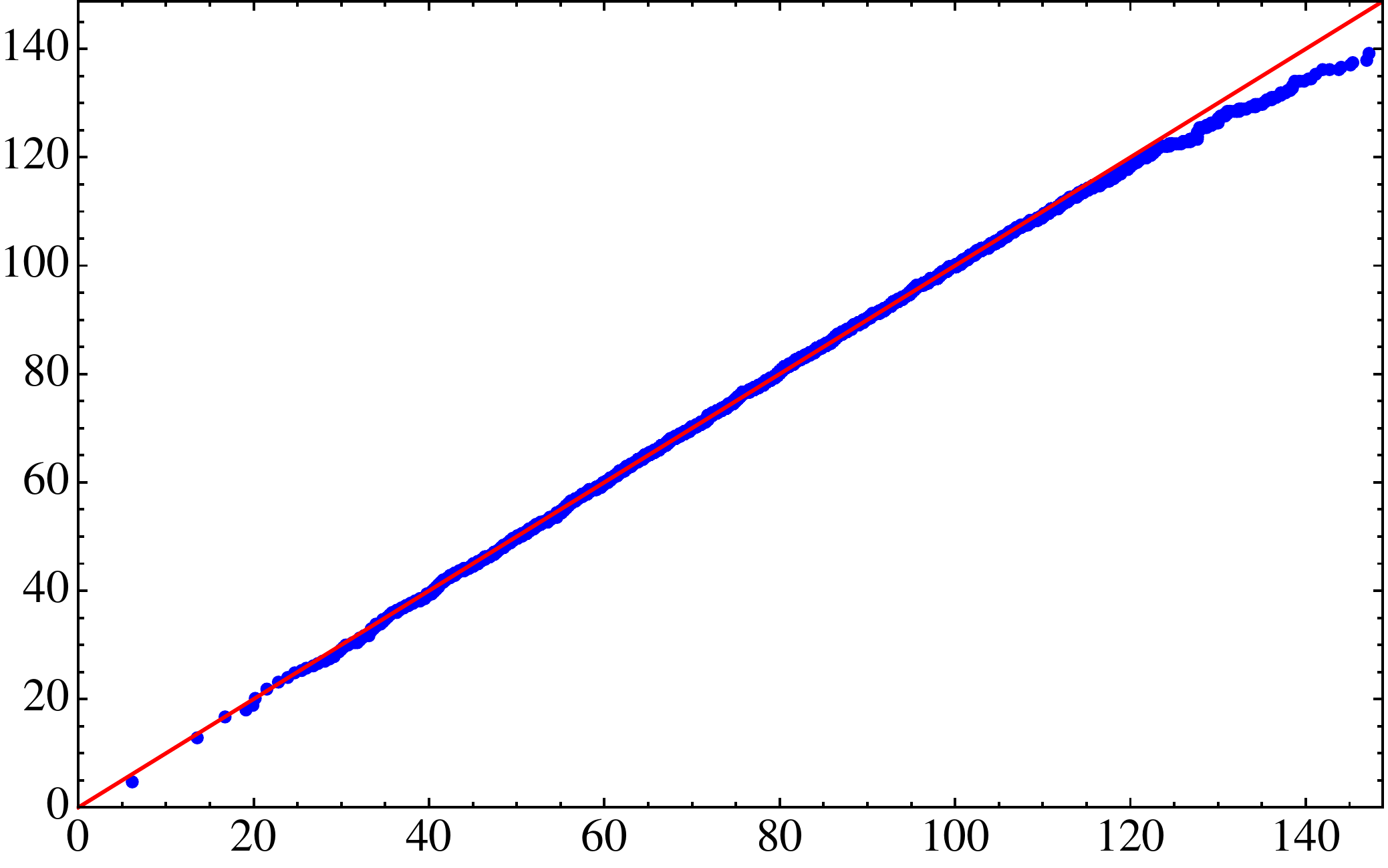}\\[3mm]
  \includegraphics[width=6.5cm]{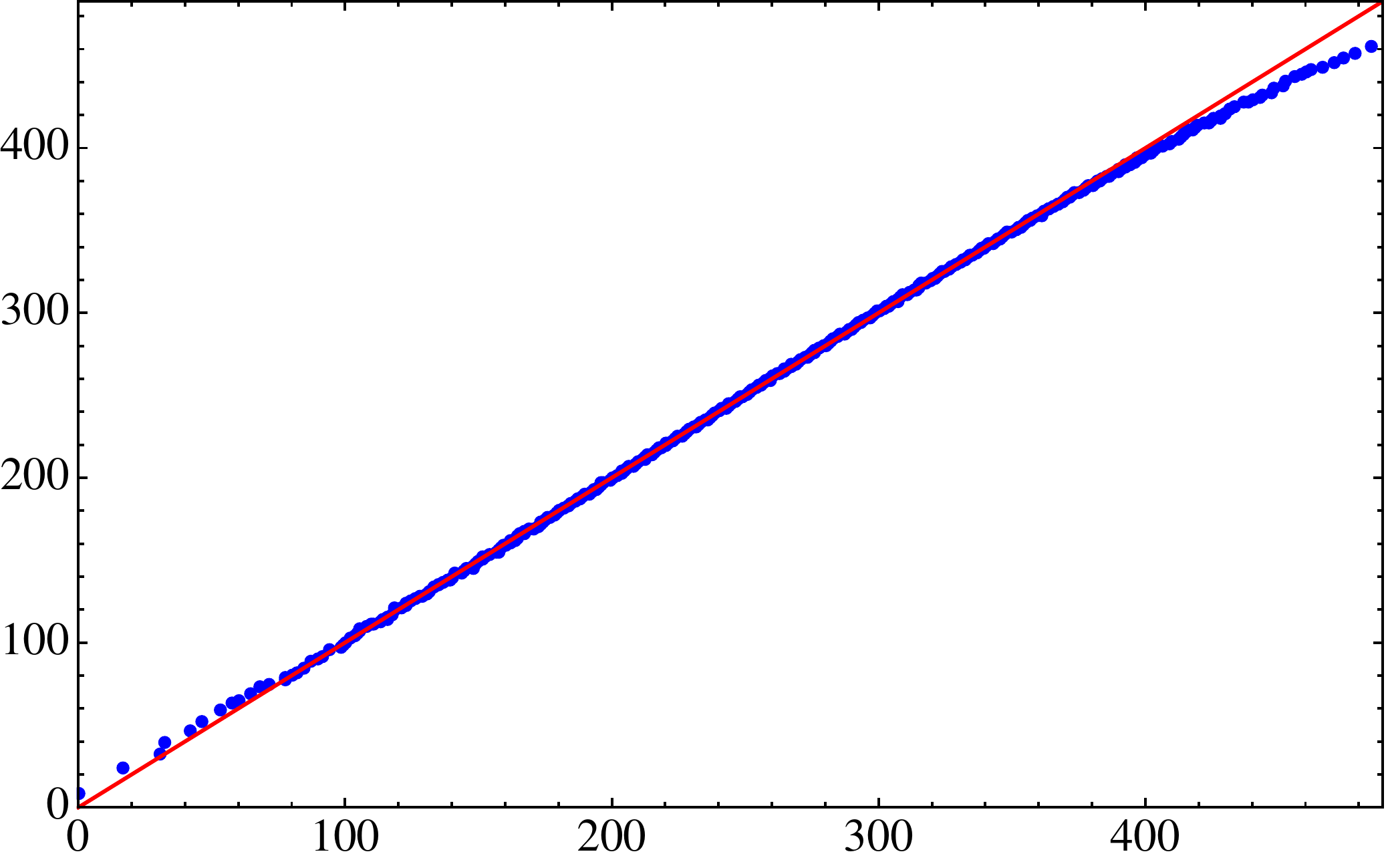}\hfill
   \includegraphics[width=6.5cm]{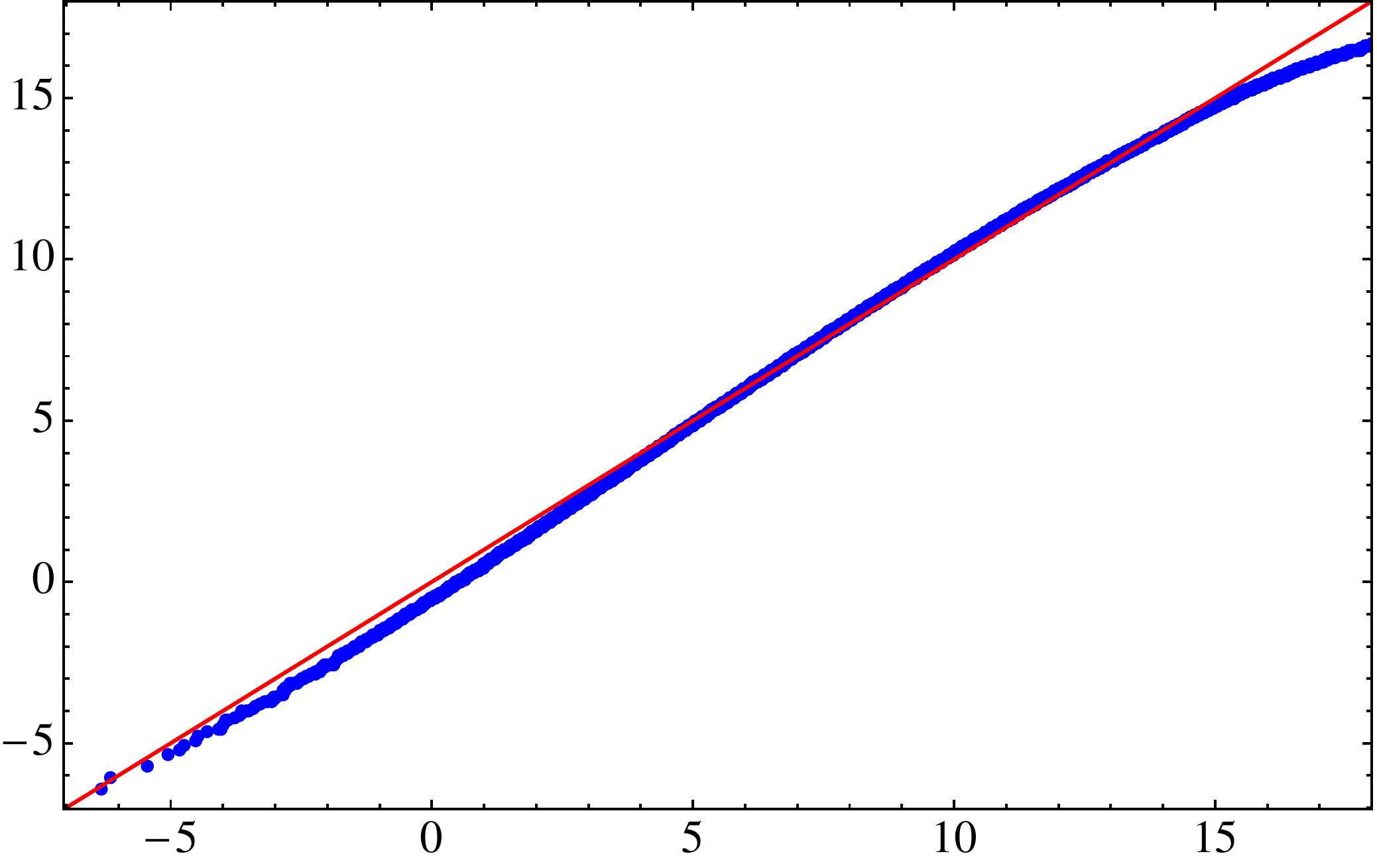}\hfill
   \caption{$q$-$q$ plots of the Gaussian distribution (with parameters $\mu$ and $\si$ computed
     from the spectrum) vs.\ the level density of Inozemtsev's chain for $N=18$, $m=2$ (top left)
     and $N=12$, $m=3$ (top right), where in both cases $k^2=1/2$. For comparison purposes, we
     have also shown the corresponding $q$-$q$ plots for the Haldane--Shastry chain with $N=18$,
     $m=2$ (bottom left) and the Heisenberg chain with $N=18$ spins (bottom right). In all four
     graphs, the straight line $y=x$ has been represented in red.}
  \label{fig.qqplots}
\end{figure}
\begin{table}[h]
  \centering
  \begin{tabular}{|c|c|c|c|c|}
    \hline
    Spin chain & $N$ & $m$ & $k^2$ & RMS error\vrule height10pt width0pt\\
    \hline\hline
               & & & 0.2 & $2.707\times 10^{-3}$\vrule height10pt width0pt\\
    \cline{4-5}
               & $18$ & $2$ & 0.5 & $2.776 \times 10^{-3}$\vrule height10pt width0pt\\
    \cline{4-5}
    \smash{\raisebox{1.7ex}{Inozemtsev}} & & & 0.8 & $3.594\times 10^{-3}$\vrule height10pt width0pt\\
    \cline{2-5}
               & $12$ & $3$ & 0.5 & $4.005 \times 10^{-3}$\vrule height10pt width0pt\\
    \hline
               & $18$ & & & $2.650\times 10^{-3}$\vrule height10pt width0pt\\
    \cline{2-2}\cline{5-5}
    \smash{\raisebox{1.4ex}{HS}} & $50$ & \smash{\raisebox{1.4ex}{$2$}} & \smash{\raisebox{1.4ex}{$(0)$}} & $5.103\times 10^{-4}$\vrule height10pt width0pt\\
    \hline
    Heisenberg & $18$ & $2$ & --- & $1.239\times 10^{-2}$\vrule height10pt width0pt\\
    \hline
  \end{tabular}
  \caption{Root mean square error of the fit of the cumulative level density of the Inozemtsev,
    Haldane-Shastry and Heisenberg chains
    to a cumulative Gaussian distribution with parameters $\mu$ and $\sigma$ taken from the
    spectrum.}\label{table.RMSE}
\end{table}

\subsection{Exact formulas for the mean and variance of the energy}
Once it has been established that for sufficiently large $N$ the level density of Inozemtsev's
elliptic chain is well approximated by a Gaussian distribution with parameters $\mu$ and $\si$
taken from the spectrum, it is of interest to evaluate the latter parameters explicitly. Although
the spectrum of Inozemtsev's chain has not been computed exactly, we shall next see that it is
possible to obtain a closed-form expression for the mean and the variance of its levels as
functions of the number of sites $N$, the number of internal degrees of freedom $m$ and the
modulus $k$. Using this expression, we shall show that (for $0<k<1$) $\mu$ and $\si^2$ grow as
$N^3$ and $N^5$, respectively, just as is the case for the Haldane--Shastry chain~\cite{FG05}.

Let us begin with the mean energy
\[
\mu=m^{-N}\tr \cH=\frac1{4m^{N}}\,\sum_{i\ne j}h_{ij}\tr(1-S_{ij})\,,
\]
where
\[
h_{ij}=\sn^{-2}\biggl(2(i-j)\frac KN\biggr)\equiv h(i-j)\,.
\]
Clearly $\tr S_{ij}=m^{N-1}$, and therefore
\[
\mu=\frac{m-1}{4m}\sum_{i\ne j}h_{ij}\,.
\]
In order to simplify the latter sum, we note that the function $h(z)\equiv\sn^{-2}(2Kz/N)$ satisfies
\begin{equation}\label{hids}
h(z)=h(-z)=h(N-z)\,,\qquad z\in\CC\,,
\end{equation}
and therefore
\[
\sum_{i\ne j}h_{ij}=\sum_{i\ne j}h(i-j)=2\sum_{l=1}^{N-1}(N-l)h(l)\,.
\]
On the other hand, performing the change of index $l\mapsto N-l$ and using the
identity~\eqref{hids} we easily obtain
\[
\sum_{l=1}^{N-1}(N-l)h(l)=\sum_{l=1}^{N-1}l\,h(l)\quad\implies\quad 2\sum_{l=1}^{N-1}l\,
h(l)=N\sum_{l=1}^{N-1}h(l)\,,
\]
so that
\begin{equation}
  \label{hsum}
  \sum_{i\ne j}h_{ij}=N\sum_{l=1}^{N-1}h(l)\,.
\end{equation}
Thus the mean energy per site is given by
\begin{equation}\label{mu}
\frac{\mu}N=\frac{m-1}{4m}\,\sum_{l=1}^{N-1}h(l)\,.
\end{equation}

Likewise, the variance of the energy can be expressed as
\[
\si^2=m^{-N}\tr\big[(\cH-\mu)^2\big]=m^{-N}\tr\big[(\cH'-\mu')^2\big]=m^{-N}\tr(\cH'^2)-\mu'^2\,,
\]
where
\begin{equation}\label{cHp}
\cH'=\frac14\sum_{i\ne j}h_{ij}S_{ij}
\end{equation}
and
\begin{equation}\label{mup}
\mu'=m^{-N}\tr\cH'=\frac{1}{4m}\sum_{i\ne j}h_{ij}=\frac{N}{4m}\,\sum_{l=1}^{N-1}h(l)\,.
\end{equation}
The trace of $\cH'^2$ can be easily computed by noting that
\[
\tr(S_{ij}S_{ln})=\begin{cases}
  m^N\,,\qquad& \{i,j\}=\{l,n\}\,,\\
  m^{N-2}\,,&\text{otherwise.}
\end{cases}
\]
We thus have
\[
16 m^{-N}\tr(\cH'^2)=m^{-N}\sum_{i\ne j,k\ne l}h_{ij}h_{kl}\tr (S_{ij}S_{kl})=
m^{-2}\Big(\sum_{i\ne j}h_{ij}\Big)^2+2(1-m^{-2})\sum_{i\ne j}h_{ij}^2\,,
\]
whence, using Eqs.~\eqref{hsum} (with $h$ replaced by $h^2$) and~\eqref{mup}, we finally obtain
\begin{equation}
  \label{si2}
  \frac{\si^2}N=\frac{m^2-1}{8Nm^2}\,\sum_{i\ne j}h_{ij}^2=
  \frac{m^2-1}{8m^2}\,\sum_{l=1}^{N-1}h(l)^2\,.
\end{equation}
It should be noted that Eqs.~\eqref{mu} and~\eqref{si2} are actually valid for any spin chain of
the form
\begin{equation}
  \label{hchain}
  \cH=\frac14\sum_{i\ne j}h_{ij}(1-S_{ij})
\end{equation}
with $h_{ij}=h(i-j)$, provided only that the function $h$ satisfies the identities~\eqref{hids}.

By Eqs.~\eqref{mu} and~\eqref{si2}, the mean and variance of the energy are proportional to the
sums
\begin{equation}
  \label{Sp}
  S_p\equiv\sum_{j=1}^{N-1}h(j)^p\equiv\sum_{j=1}^{N-1}\sn^{-2p}\bigg(\frac{2jK}N\bigg)
\end{equation}
with $p=1,2$. In the trigonometric case ($k=0$), these sums have been evaluated in closed form (in
terms of Bernoulli numbers) by Berndt and Yeap~\cite{BY02}. They can be easily computed for
arbitrary $k\in[0,1]$ using the results in~\ref{app.sums}, as we shall next show.

Consider, to begin with, the sum $S_1$. Using Eqs.~\eqref{snwp} and~\eqref{S1p} with
\[
\om_3=\frac{\iu NK'}{2K}\equiv \frac{\iu N\tau}2
\]
we easily obtain
\[
S_1=2\,\bigg(\frac{N}{2K}\bigg)^2\big[\eta_1(1/2,\iu N\tau/2)-\eta_1(N/2,\iu
N\tau/2)\big]+\frac13(N-1)(1+k^2)\,,
\]
where $\eta_1$ is defined in Eq.~\eqref{etai}. From the homogeneity relation
\[
\eta_1(N/2,\iu N\tau/2)=\frac{2K}N\,\eta_1(K,\iu K')
\]
(cf.~Eq.~\eqref{hometai}) and the well-known identity~\cite[Eq.~18.9.13]{AS70}
\begin{equation}\label{eta1EK}
\eta_1(K,\iu K')=E-\frac13\,(2-k^2)K\,,
\end{equation}
where
\[
E(k)=\int_0^{\pi/2}\sqrt{1-k^2\sin^2\vp}\,\diff\vp
\]
is the complete elliptic integral of the second kind, we finally obtain
\begin{equation}
  \label{S1final}
  S_1=\frac{N^2}{2K^2}\,\eta_1(1/2,\iu N\tau/2)+N\bigg(1-\frac{E}{K}\bigg)-\frac13\,(1+k^2)\,.
\end{equation}
Likewise, from Eqs.~\eqref{snwp},~\eqref{eta1EK},~\eqref{S1p},~\eqref{wp2sum} (with $\om_3=\iu
N\tau/2$), and the identity
\[
\bigg(\frac{N}{2K}\bigg)^4g_2(N/2,\iu N\tau/2)=g_2(K,\iu K')=\frac43\,(k^4-k^2+1)
\]
(cf.~\ref{app.prel} and~\cite[Eq.~18.9.4]{AS70}), after a straightforward calculation we obtain
\begin{align}\
  S_2=&\frac1{960}\,\bigg(\frac{N}{K}\bigg)^4 g_2(1/2,\iu
  N\tau/2)+\frac13\,(1+k^2)\,\bigg(\frac{N}{K}\bigg)^2\eta_1(1/2,\iu
  N\tau/2)+\frac{N}3\,\bigg(2+k^2-2(1+k^2)\frac{E}{K}\bigg)\notag\\
  &-\frac1{45}\,(11k^4+4k^2+11)\,.
  \label{S2final}
\end{align}
Substituting Eqs.~\eqref{S1final}-\eqref{S2final} into~\eqref{mu}-\eqref{si2} we finally obtain
the following closed-form expressions for the mean $\mu$ and the variance $\si^2$ of the energy
per site of the $\mathrm{su}(m)$ Inozemtsev chain~\eqref{cH}:
\begin{align}
  \label{mufinal}
  \frac{\mu}{N}&=\frac{m-1}{4m}\,\bigg[\frac{N^2}{2K^2}\,\eta_1(1/2,\iu
  N\tau/2)+N\,\bigg(1-\frac{E}{K}\bigg)-\frac13\,(1+k^2)\bigg]\\
  \frac{\si^2}{N}&=\frac{m^2-1}{8m^2}\,\bigg[\frac1{960}\,\bigg(\frac{N}{K}\bigg)^4 g_2(1/2,\iu
  N\tau/2)+\frac13\,(1+k^2)\,\bigg(\frac{N}{K}\bigg)^2\eta_1(1/2,\iu
  N\tau/2)\notag\\
  &\hphantom{=\frac{m^2-1}{8m^2}\,\bigg[}\enspace+\frac{N}3\,\bigg(2+k^2-2(1+k^2)
  \frac{E}{K}\bigg)-\frac1{45}\,(11k^4+4k^2+11)\bigg]\,.
  \label{si2final}
\end{align}
It is instructive to check that when $k=0$ the above formulas reduce to the well-known ones for
the Haldane--Shastry chain derived in Ref.~\cite{FG05}. To this end, we shall make use of the
asymptotic expansions
\begin{equation}\label{eta1g2as}
\eta_1(1/2,\om_3)=\pi^2\bigg(\frac1{6}-4\sum_{k=1}^\infty\si_1(k)\,q^{2k}\bigg)\,,\qquad
g_2(1/2,\om_3)=20\pi^4\bigg(\frac1{15}+16\sum_{k=1}^\infty\si_3(k)\,q^{2k}\bigg)\,,
\end{equation}
where $\Im\om_3>0$, $q=\e^{2\pi\iu\om_3}$, and $\si_r(n)$ denotes the divisor function\footnote{By
  definition, $j|n$ if $n$ is a multiple of $j$.}
\[
\si_r(n)=\sum_{j|n}j^r
\]
(see, respectively,~\cite[Eq.~23.8.5]{OLBC10} and~\cite{SS03}). Thus
\begin{equation}\label{eta1g2lim}
  \eta_1(1/2,+\iu\infty)=\frac{\pi^2}6\,,
  \qquad g_2(1/2,+\iu\infty)=\frac43\,\pi^4\,,
\end{equation}
and therefore, taking into account that
\[
K(0)=\pi/2\,,\en K'(0)=+\infty\en\implies\en \tau(0)=+\infty\,,
\]
from Eqs.~\eqref{mufinal}-\eqref{si2final} with $k=0$ we readily obtain the well-know expressions
for the mean and variance of the energy per site of the Haldane--Shastry chain:
\[
\frac{\mu}N=\frac{m-1}{12m}\,(N^2-1)\,,\qquad\frac{\si^2}N=\frac{m^2-1}{360m^2}\,(N^2-1)(N^2+11)\,.
\] {}From the latter formulas it follows that for the Haldane--Shastry chain $\mu$ and $\si^2$
grow with $N$ respectively as $N^3$ and $N^5$. That this is the case for arbitrary $k\in[0,1)$ can
be easily shown from the asymptotic formula~\eqref{Sppasymp} and the identity~\eqref{snwp}, which
yield
\begin{align}\label{muasymp1}
  \frac{4m\ms\mu}{N(m-1)}&=S_1=\bigg(\frac{N}{2K}\bigg)^2\frac{(2\pi)^2}{2!}\,\frac16+O(N)=\frac1{12}\bigg(\frac{N\pi}{K}\bigg)^2+O(N)\,,\\
  \frac{8m^2\ms\si^2}{N(m^2-1)}&=S_2=\bigg(\frac{N}{2K}\bigg)^4\frac{(2\pi)^4}{4!}\,\frac1{30}+O(N^3)=\frac1{720}\bigg(\frac{N\pi}{K}\bigg)^4+O(N^3)\,.
\label{si2asymp1}
\end{align}

\begin{figure}[h]
\includegraphics[height=4.2cm]{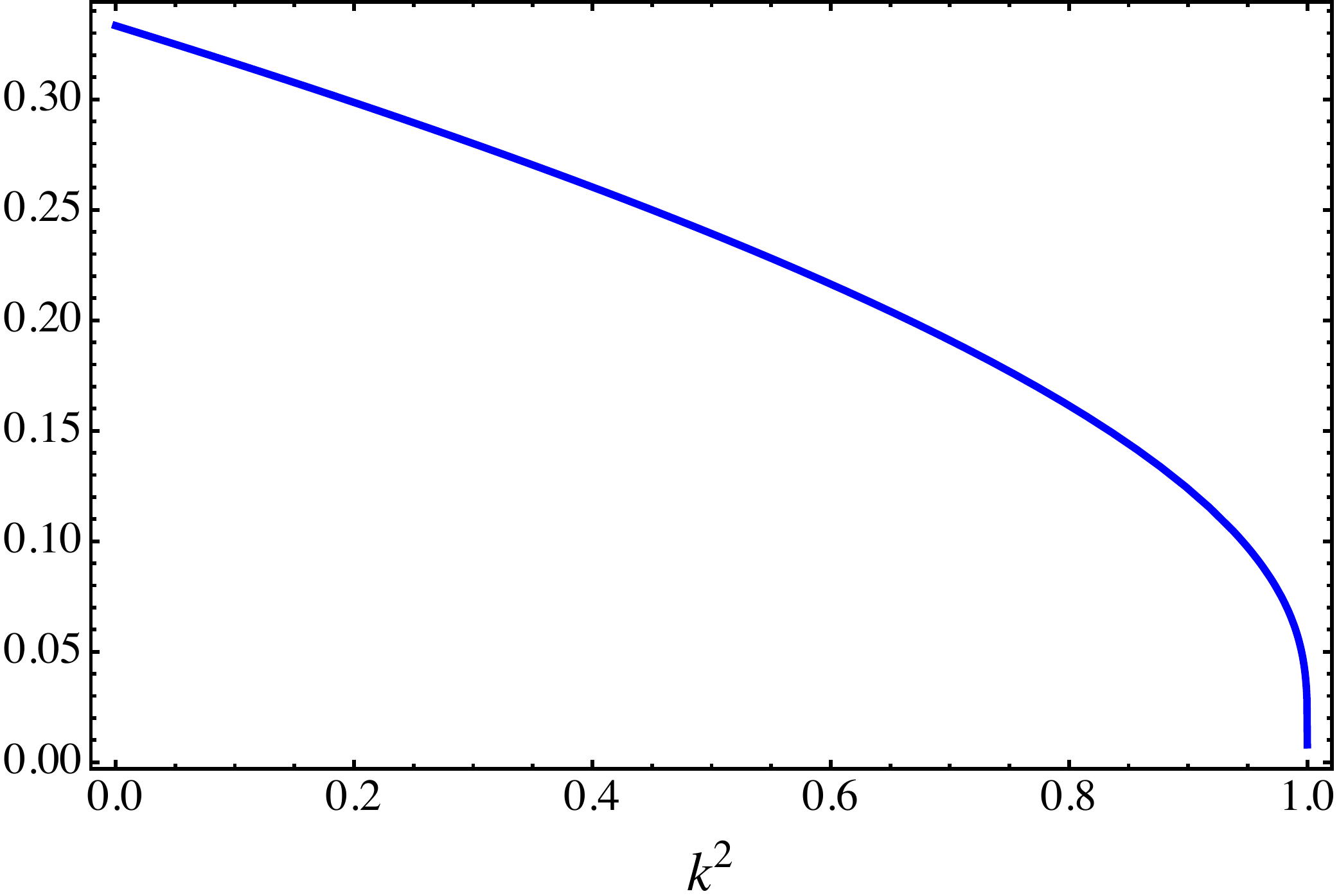}\hfill 
\includegraphics[height=4.2cm]{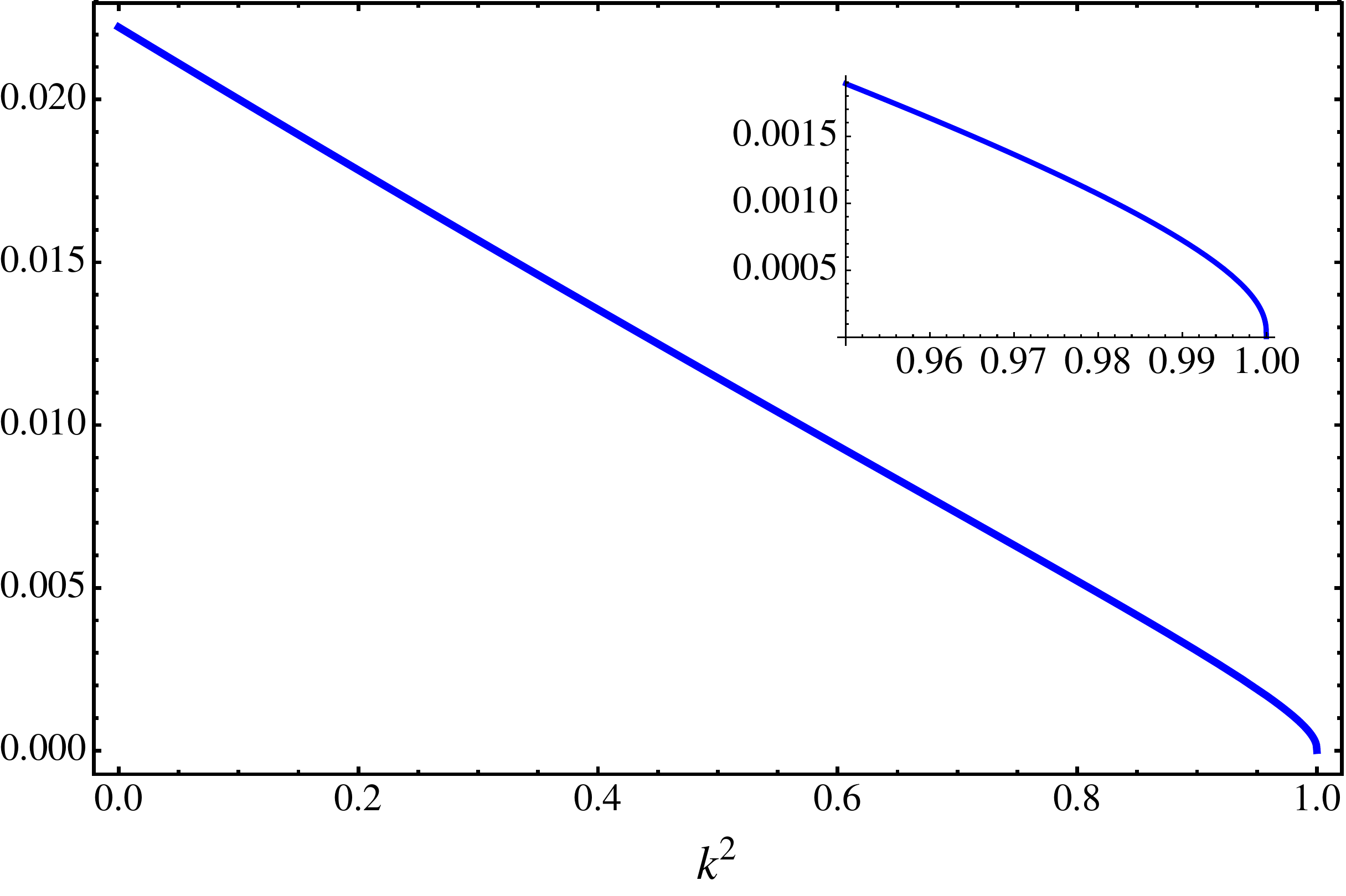}
\caption{Coefficients of $N^3$ (left) and $N^5$ (right) in the asymptotic expansions of
  $4\,m\,\mu/(m-1)$ and $8\,m^2\si^2/(m^2-1)$, respectively. Inset: behavior of the latter
  coefficient for $0.95\le k^2\le 1$.}
\label{fig.coeffs}
\end{figure}
Much sharper asymptotic formulas can be obtained using Eqs.~\eqref{eta1g2as} with $\om_3=\iu
N\tau/2\to+\iu\infty$, namely
\begin{align}
  \label{muasymp}
  \frac{4m\ms\mu}{N(m-1)}&=\frac1{12}\bigg(\frac{N\pi}{K}\bigg)^2+N\,\bigg(1-\frac{E}{K}\bigg)-\frac13\,(1+k^2)+O\big(N^2\e^{-2N\pi\tau}\big)\,,\\
  \frac{8m^2\ms\si^2}{N(m^2-1)}&=\frac1{720}\bigg(\frac{N\pi}{K}\bigg)^4+\frac1{18}\,(1+k^2)\,\bigg(\frac{N\pi}{K}\bigg)^2+\frac{N}3\,\bigg(2+k^2-2(1+k^2)
  \frac{E}{K}\bigg)\notag\\
  &\hphantom{=\frac1{720}\bigg(\frac{N\pi}{K}\bigg)^4}-\frac1{45}\,(11k^4+4k^2+11)+O\big(N^4\e^{-2N\pi\tau}\big)\,.
  \label{si2asymp}
\end{align}
In Figure~\ref{fig.coeffs} we have plotted the coefficients of $N^3$ and $N^5$ in the asymptotic
expansions of $4m\mu/(m-1)$ and $8m^2\si^2/(m^2-1)$ as functions of $k^2$. Since $K(1)=+\infty$,
both coefficients vanish as $k\to1$, a fact that is also apparent from the latter figure. This was
to be expected, since when $k\to1$ the Hamiltonian~\eqref{cH} tends to~\eqref{hchain} with
$h_{ij}=1$, so that both $4m\,\mu/(m-1)$ and $8m^2\si^2/(m^2-1)$ tend to $N(N-1)$ on account of
Eqs.~\eqref{mu}-\eqref{si2}. Note, however, that the asymptotic
formulas~\eqref{muasymp1}-\eqref{si2asymp1} and~\eqref{muasymp}-\eqref{si2asymp} are \emph{not}
valid when $k=1$, since $\tau$ vanishes in this case. In fact, it is straightforward to show that
Eqs.~\eqref{mufinal}-\eqref{si2final} yield the correct values of $\mu$ and $\si^2$ in the limit
$k\to 1$ with the help of the formulas
\[
\lim_{\om_3\to0}\big[(2\om_3)^2\eta_1(1/2,\om_3)\big]=\frac{\pi^2}6\,,\qquad
\lim_{\om_3\to0}\big[(2\om_3)^4g_2(1/2,\om_3)\big]=\frac43\,\pi^4\,,
\]
which can be easily derived from~\eqref{eta1g2lim} using the homogeneity properties of $\eta_1$
and $g_2$ and Legendre's relation~\eqref{Legendre}.

\subsection{Skewness and excess kurtosis}

Given a probability distribution $p(E)$ with mean $\mu$ and standard deviation $\si$, its
\emph{skewness} $\ga_1$ and (excess) \emph{kurtosis} $\ga_2$ are defined by
\[
\ga_1=\frac{\big\langle(E-\mu)^3\big\rangle}{\si^3}\,,\qquad
\ga_2=\frac{\big\langle(E-\mu)^4\big\rangle}{\si^4}-3\,,
\]
where the average $\big\langle f(E)\big\rangle$ is given by
\[
\big\langle f(E)\big\rangle=\int_{-\infty}^\infty f(E)p(E)\,\diff E\,.
\]
Since a normal distribution with arbitrary parameters $\mu$ and $\si$ has zero skewness and
kurtosis, a necessary condition for the level density of the Inozemtsev chain~\eqref{cH} to be
asymptotically normal is the vanishing as $N\to\infty$ of the skewness and kurtosis of its
spectrum, given by
\begin{equation}\label{skewkurt}
\ga_1=\frac{\tr\big[(\cH-\mu)^3\big]}{m^N\si^3}\,,\qquad
\ga_2=\frac{\tr\big[(\cH-\mu)^4\big]}{m^N\si^4}-3\,.
\end{equation}
In the rest of this section we shall
rigorously prove that this is indeed the case (cf.~Eqs.~\eqref{ga1est}
and~\eqref{ga2est}).

As a matter of fact, we shall compute the skewness and kurtosis of the spectrum of a general chain
of the form~\eqref{hchain}, where $h_{ij}=h(i-j)$ and $h$ satisfies the identities~\eqref{hids}.
Note, to begin with, that we obviously have
\begin{equation}\label{ga1ga2}
\ga_1=\frac{\tr\big[(\mu'-\cH')^3\big]}{m^N\si^3}\,,\qquad
\ga_2=\frac{\tr\big[(\cH'-\mu')^4\big]}{m^N\si^4}-3\,,
\end{equation}
where $\cH'$ and~$\mu'$ are respectively given by Eqs.~\eqref{cHp} and~\eqref{mup}. In the case of
the skewness, expanding $(\mu'-\cH')^3$ in powers of $\cH'$ and taking into account that
\[
m^{-N}\tr\big(\cH'^2\big)=\mu'^2+\si^2
\]
we easily obtain
\begin{equation}\label{ga1def}
\ga_1=\frac1{\si^3}\,\big(-m^{-N}\tr(\cH'^3\big)+3\mu'\si^2+\mu'^3\big)\,.
\end{equation}
By Eq.~\eqref{cHp}, in order to evaluate the trace of $\cH'^3$ we must compute the trace of the
product of three spin permutation operators. It is not difficult to show that
\[
\tr(S_{ij}S_{ln}S_{pq})=
\begin{cases}
  m^{N-1}\,,&\{i,j\}=\{l,n\}\en\text{or}\en\{i,j\}=\{p,q\}\en\text{or}\en\{l,n\}=\{p,q\}\\
  &\text{or}\en\card\{i,j,l,n,p,q\}=3\\
  m^{N-3}\,,&\text{otherwise}\,,
\end{cases}
\]
which after a long but straightforward calculation leads to the identity
\begin{multline*}
  64m^{3-N}\tr(\cH'^3)= \bigg(\sum_{i\ne j}h_{ij}\bigg)^3+6\,(m^2-1)\bigg(\sum_{i\ne
    j}h_{ij}^2\bigg)\bigg(\sum_{l\ne n}h_{ln}\bigg)+8\,(m^2-1)\bigg(
  {\sum_{i,j,l}}'h_{ij}h_{jl}h_{li}-\sum_{i\ne j}h_{ij}^3\bigg)\,.
\end{multline*}
Here, as in the rest of the paper, we have denoted by $\sum\nolimits'$ a sum in which no pair of
indices can take the same value. Combining the previous formula with Eqs.~\eqref{mup},
\eqref{si2}, and~\eqref{ga1def} we obtain the following explicit expression for the skewness of
the spectrum of a general spin chain of the form~\eqref{hchain}:
\begin{equation}
  \label{ga1}
  \ga_1=\frac{m^2-1}{8m^3\si^3}\,\bigg(\sum_{i\ne j}h_{ij}^3-\sump{i,j,l}h_{ij}h_{jl}h_{li}\bigg)\,.
\end{equation}
Using Eqs.~\eqref{hsum} (with $h$ replaced by $h^3$) and~\eqref{sumskew} to rewrite the last sum we
arrive at the  simplified expression
\begin{equation}
  \label{ga1final}
  \ga_1=\frac{m^2-1}{8m^3\si^3}\,N\,\bigg(S_3-2\sum_{1\le i<j\le N-1}h(i)\ms
  h(j)\ms h(j-i)\bigg)\,.
\end{equation}

Although the previous formula holds for any chain of the form~\eqref{hchain}, the behavior of
$\ga_1$ when $N\to\infty$ depends on the specific properties of the function $h$. In the case of
the Inozemtsev chain~\eqref{cH} the function $\sn x$ is monotonically increasing in the interval
$[0,K]$ and is symmetric about $K$, so that
\[
h(l)=\sn^{-2}\biggl(\frac{2lK}N\biggr)\le\sn^{-2}\biggl(\frac{2K}N\biggr)\,,\qquad 1\le l\le N-1\,.
\]
Hence
\begin{equation}\label{hihjest}
\sum_{1\le i<j\le N-1}h(i)\ms h(j)\ms h(j-i)\le \sn^{-2}\bigl(\tfrac{2K}N\bigr)
\sum_{1\le i<j\le N-1}h(i)\ms h(j)
=\frac12\,\sn^{-2}\bigl(\tfrac{2K}N\bigr)\,(S_1^2-S_2)\,,
\end{equation}
and using this estimate in Eq.~\eqref{ga1final} we obtain
\begin{equation}\label{gamma1est}
  -\frac{m^2-1}{8m^3}\,\frac{N}{\si^3}\,(S_1^2-S_2)\sn^{-2}\bigl(\tfrac{2K}N\bigr)\le \ga_1\le\frac{m^2-1}{8m^3}\,\frac{NS_3}{\si^3}\,.
\end{equation}
Since
\begin{equation}\label{SpN2p}
S_p\sim N^{2p}
\end{equation}
as $N\to\infty$ by Eqs.~\eqref{wpsn}~and~\eqref{Sppasymp}, and moreover
\[
\sn^{-2}\bigl(\tfrac{2K}N\bigr)\sim N^2\,,\qquad
\si^3\sim N^{15/2}\,,
\]
from~\eqref{gamma1est} we conclude that
\begin{equation}\label{ga1est}
\ga_1=O(N^{-1/2})\underset{N\to\infty}{\longrightarrow}0\,.
\end{equation}
In fact, our numerical calculations strongly suggest that $\ga_1$ tends to zero much faster than
$N^{-1/2}$ as $N\to\infty$, namely\footnote{By definition, $f(N)\sim N^p$ if $N^{-p}f(N)$ has a
  finite and non-vanishing limit as $N\to\infty$.}
\[
\ga_1\sim N^{-5/2}\,.
\]

\smallskip
Consider next the kurtosis~$\ga_2$. Expanding $(\cH'-\mu')^4$ in powers of $\cH'$ in the second
Eq.~\eqref{ga1ga2} and taking into account Eq.~\eqref{ga1def} we obtain
\begin{equation}
  \label{kurtdef}
  \ga_2=\frac1{\si^4}\,\big(m^{-N}\tr(\cH'^4)+4\ga_1\mu'\si^3-6\mu'^2\si^2-\mu'^4\big)-3\,.
\end{equation}
An analysis similar to the above for $\ga_1$, but considerably more involved, yields the following
expression for the kurtosis of the general chain~\eqref{hchain}:
\begin{equation}
  \label{ga2}
  \ga_2=\frac{m^2-1}{16m^4\si^4}\,\bigg((3-m^2)\sum_{i\ne j}h_{ij}^4-m^2\sump{i,j,l}
  h_{ij}^2h_{jl}^2+2(m^2-6)\sump{i,j,l} h_{ij}^2h_{jl}h_{li}+
  3\sump{i,j,l,n} h_{ij}h_{jl}h_{ln}h_{ni}\bigg)\,.
\end{equation}
Using Eqs.~\eqref{sum2}--\eqref{sum45} we can rewrite this equation more compactly as
\begin{multline}
   \ga_2=\frac{m^2-1}{16m^4\si^4}\,N\,\bigg(3S_4-m^2S_2^2+4(m^2-6)\sum_{1\le i<j\le
     N-1}h^2(i)\ms h(j)\ms h(j-i)\\
   +6\sum_{3\le i+j+l\le N-1}h(i)\ms h(j)\ms h(l)h(i+j+l)
   +12\sum_{3\le i+j+l\le N-1}h(i)\ms h(j)\ms h(i+l)h(j+l)\bigg)\,.
 \label{ga2simp}
\end{multline}
Each of the terms in parentheses in the latter equation is $O(N^8)$. Indeed, for the first two
terms this is an immediate consequence of Eq.~\eqref{SpN2p}. As for the third one, we have
\begin{align}
  \sum_{1\le i<j\le N-1}h^2(i)\ms h(j)\ms h(j-i)&=\frac12\sum_{1\le i<j\le
    N-1}\big[h^2(i)\ms h(j)+h^2(j)\ms h(i)\big]\,\ms h(j-i)\nonumber\\
  &=\frac12\sum_{1\le i\ne j\le N-1}h^2(i)\ms h(j)\ms
  h(j-i)\le\frac12\sn^{-2}\bigl(\tfrac{2K}N\bigr)\sum_{1\le i\ne j\le N-1}h^2(i)\ms h(j)\nonumber\\
  &=\frac12\sn^{-2}\bigl(\tfrac{2K}N\bigr)\,(S_1S_2-S_3)=O(N^8)\,.
  \label{sum3est}
\end{align}
The second sum in Eq.~\eqref{ga2simp} can be easily estimated by noting that
\[
\sum_{3\le i+j+l\le N-1}h(i)\ms h(j)\ms h(l)\ms h(i+j+l)\le
\sn^{-2}\bigl(\tfrac{2K}N\bigr)\sum_{3\le i+j+l\le N-1}h(i)\ms h(j)\ms h(l)
\]
and
\begin{multline*}
\sum_{3\le i+j+l\le N-1}h(i)\ms h(j)\ms h(l)=\sum_{\substack{3\le i,j,l\le N-1\\ i+j+l\ge
    2N+1}}h(i)\ms h(j)\ms h(l)\en
\implies\en\sum_{3\le i+j+l\le N-1}h(i)\ms h(j)\ms h(l)\le\frac12\,S_1^3\,,
\end{multline*}
so that
\begin{equation}
  \label{sum4est}
  \sum_{3\le i+j+l\le N-1}h(i)\ms h(j)\ms h(l)\ms h(i+j+l)\le
  \frac12\,\sn^{-2}\bigl(\tfrac{2K}N\bigr)\,S_1^3=O(N^8)\,.
\end{equation}
Likewise,
\begin{align}
  \sum_{3\le i+j+l\le N-1}&h(i)\ms h(j)\ms h(i+l)h(j+l)\le\sn^{-2}\bigl(\tfrac{2K}N\bigr)
  \sum_{3\le i+j+l\le N-1}h(i)\ms h(j)\ms h(i+l)\nonumber\\&=
  \sn^{-2}\bigl(\tfrac{2K}N\bigr)
  \sum_{\substack{1\le i,j,l\le N-1\\i<l,\ j+l\le N-1}}h(i)\ms h(j)\ms h(l)\le
  \frac12\,\sn^{-2}\bigl(\tfrac{2K}N\bigr)\,S_1^3=O(N^8)\,.
  \label{sum5est}  
\end{align}
From the estimates~\eqref{sum3est}--\eqref{sum5est} and the fact that $\si^4\sim N^{10}$ it immediately
follows that
\begin{equation}
  \label{ga2est}
  \ga_2=O(N^{-1})\underset{N\to\infty}{\longrightarrow}0\,.
\end{equation}
In fact, numerical evaluation of $\ga_2$ for $N$ up to $300$ and different values of $k^2$ shows
that
\[
\ga_2\sim N^{-1}
\]
as $N\to\infty$.

\section{Integrability in terms of statistical properties of the spectrum}\label{sec.int}

In this section we shall study several global properties of the spectrum of the Inozemtsev chain
which are of interest for analyzing the integrable versus chaotic behavior of a quantum system. As
is well known, a preliminary step in the analysis of a spectrum in the context of quantum chaos is
the computation of a smooth approximation to its (cumulative) level density. This is essentially
due to the fact that, in order to compare different spectra (or different parts of a given
spectrum), it is necessary to first normalize them so that the resulting spectra have an
approximately uniform level density. More precisely, let $E_0<E_1<\cdots<E_n$ be the distinct
energy levels of a (finite) quantum spectrum, with respective degeneracies $d_0,\dots,d_n$. If
$\vep(E)$ is a smooth approximation to the cumulative level density
\[
F(E)=\frac1D\sum_{i;E_i\le E}d_i\,,\qquad D\equiv\sum_{i=0}^nd_i\,,
\]
it can be easily shown~\cite{Ha01} that the level density of the ``unfolded'' spectrum
\[
\vep_i\equiv\vep(E_i)\,,\qquad i=0,\dots,n\,,
\]
is approximately equal to 1.

\subsection{Spacings distribution}

The first property that we shall consider is the distribution of normalized spacings between
consecutive levels of the unfolded spectrum, given by
\[
s_i=\frac{n(\vep_{i}-\vep_{i-1})}{\vep_n-\vep_0}\,,\qquad i=1,\dots,n\,.
\]
According to a well-known conjecture due to Bohigas, Giannoni and Schmit~\cite{BGS84}, in a
quantum system invariant under time reversal (with integer total spin, or invariant under
rotations around an axis) whose classical limit is chaotic the probability density $p(s)$ of these
consecutive spacings should approximately follow Wigner's law
\begin{equation}\label{wigner}
p(s)=\frac{\pi s}2\ms\e^{-\pi\ms s^2/4}\,,
\end{equation}
characteristic of the Gaussian orthogonal ensemble (GOE) in random matrix theory~\cite{Me04}. In
the integrable case, through a heuristic argument based on the semiclassical limit Berry and Tabor
conjectured that the probability density $p(s)$ should be Poissonian for a generic integrable
system~\cite{BT77}. This conjecture has been confirmed for a wide class of such systems of
physical interest, including the Heisenberg chain, the Hubbard model, the $t$-$J$
model~\cite{PZBMM93} and the chiral Potts model~\cite{AMV02}. On the other hand, in most chains of
HS type (including the original Haldane--Shastry and the Polychronakos--Frahm chains) the raw
spectrum is either exactly or approximately equispaced~\cite{BFGR08,BFGR08epl,BFG09}. This fact,
together with the Gaussian character of the level density, can be shown to imply~\cite{BFGR08epl}
that the cumulative spacings distribution $P(s)=\int_0^s p(t)\,\diff t$ of the \emph{whole}
spectrum is approximately of the form
\begin{equation}\label{Ps}
P(s)=1-\frac2{\sqrt\pi s_{\mathrm{max}}}\sqrt{\log\bigg(\frac{s_{\mathrm{max}}}{s}\bigg)}\,.
\end{equation}

The energies of Inozemtsev's chain~\eqref{cH} are not even approximately equispaced, so that it
should not be expected that the spacings distribution of its whole spectrum obey the latter
equation; in fact, it can be numerically verified that this is not the case. We shall therefore
follow a more conventional approach, studying the spectrum of the restriction of the Hamiltonian
to simultaneous eigenspaces of a suitable set of mutually commuting first integrals. We shall
restrict ourselves to the $\mathrm{su}(2)$ case, for which these first integrals can be taken as
the operators $\bS^2\equiv (S^x)^2+(S^y)^2+(S^z)^2$, $S^z$ and $T$, where
\[
S^{\al}=\frac12\sum_i\si_i^\al\qquad (\al=x,y,z)\,,
\]
$\si_i^\al$ denotes the Pauli matrix $\si^\al$ acting on the internal space of the $i$-th spin,
and $T$ is the translation operator along the chain, defined by
\[
T\ket{s_1,s_2,\dots,s_N}=\ket{s_2,s_3,\dots,s_N,s_1}\,.
\]
Since $T$ is obviously unitary and satisfies $T^N=1$, its eigenvalues are the $N$ roots of unity.
Thus the eigenstates of $T$ have well-defined (modulo $2\pi$) total momentum
$P_{\mathrm {tot}}=2j\pi/N$, with $j=0,1,\dots, N-1$. Note also that in the $\mathrm su(2)$ case
Eq.~\eqref{Sij} yields
\[
S_{ij}=\frac12(1+\bsi_i\cdot\bsi_j)\,,\qquad\bsi_i\equiv(\si_i^x,\si_i^y,\si_i^z)\,,
\]
from which it easily follows that the square of the total spin operator is given by
\[
\bS^2=\sum_{i<j}S_{ij}-\frac N4\,(N-4)\,.
\]

We have numerically computed the cumulative distribution $P(s)$ of Inozemtsev's chain~\eqref{cH}
for several values of $k$ and $N$, in subspaces with zero momentum and $S^z=0$ (for even $N$) or
$S^z=1/2$ (for odd~$N$), choosing the eigenvalue of the operator $\bS^2$ so as to obtain the
largest possible spectrum\footnote{Interestingly, the level density is also approximately Gaussian
  when the Hamiltonian is restricted to (sufficiently large) subspaces with well-defined
  eigenvalues of the operators $\bS^2$, $S^z$ and $P_{\tot}$, as is the case with the whole
  spectrum (cf.~Section~\ref{sec.LD}). Thus, in all the subspaces considered the unfolding
  function $\vep(E)$ can be simply taken as the cumulative distribution function~\eqref{CGDF} of a
  normal distribution with parameters equal to the mean and the standard deviation of the
  corresponding subspectrum.}. It is apparent in all cases that approximately 25\% of the spacings
are several orders of magnitude larger than the rest, indicating that when $0<k<1$ the highly
degenerate and approximately equispaced spectrum of the Haldane--Shastry chain ($k=0$) splits into
multiple subspectra whose mutual distance is much larger than the typical spacing within each
subspectrum. When this fact is taken into account by removing the largest spacings and
renormalizing the remaining ones to mean $1$, the renormalized spacings distribution is very well
approximated by Poisson's law~(cf.~Fig.~\ref{fig.spacings} for $k^2=1/2$). Thus the spacings
distribution of Inozemtsev's chain behaves as predicted by the Berry--Tabor conjecture for a
``generic'' integrable system like, e.g., the Heisenberg chain.

\begin{figure}[ht]
  \centering
\includegraphics[width=8cm]{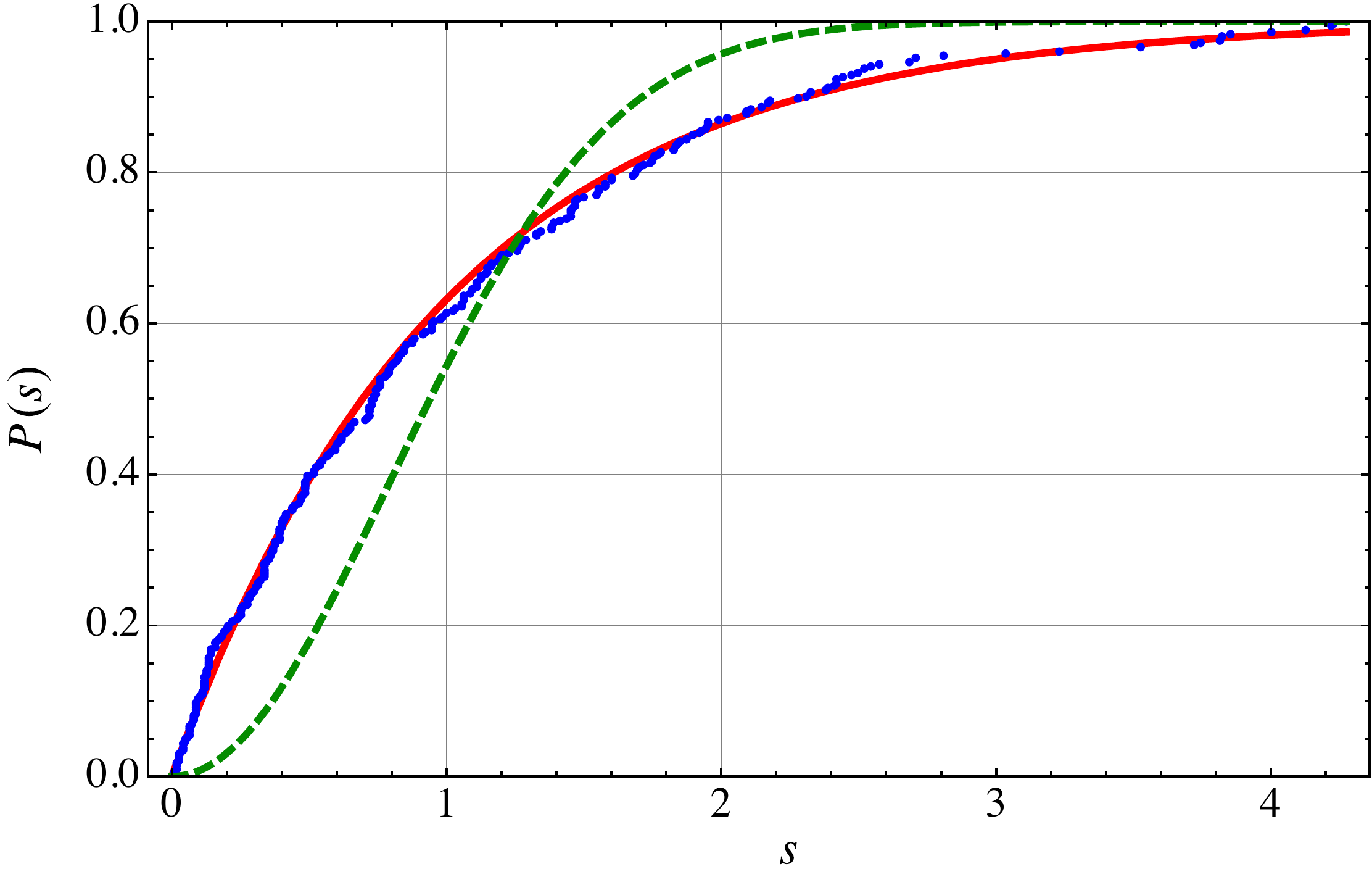}
\caption{Cumulative spacings distribution $P(s)$ of Inozemtsev's chain with $k^2=1/2$, $m=2$ and
  $N=18$ in the subspace with total spin $S=2$, $S_z=0$, $P_{\tot}=0$ after removing the largest
  25\% spacings (blue dots), compared to Poisson's law (continuous red line) and Wigner's law
  (dashed green line). The RMS error of the fit of $P(s)$ to Poisson's distribution is
  $1.33\cdot 10^{-2}$.}
\label{fig.spacings}
\end{figure}%

\subsection{Spectral noise}

In recent years, a test for detecting quantum integrability versus chaotic behavior
by analyzing a different characteristic of the spectral statistics
has been proposed by Rela\~no et al.~\cite{RGMRF02,FGMMRR04}. The test is based on the Fourier
analysis of the fluctuations of the spacings of the unfolded spectrum from their mean (the
so-called ``spectral noise''). More precisely, one considers the statistic
\[
\de_j=\sum_{i=1}^j(s_i-1)\,,\qquad j=1,\dots,n\,,
\]
and its discrete Fourier transform
\[
\hde_\nu=\frac1{\sqrt n}\sum_{j=1}^n\de_j\,\e^{-2\pi\iu j\nu/n}\,,\qquad \nu=1,\dots,n\,.
\]
It was conjectured in Ref.~\cite{RGMRF02} that the power spectrum
$\cP(\nu)\equiv\big|\hde_\nu\big|^2$ of an integrable system should behave as $\nu^{-2}$ (blue
noise) for sufficiently small values of $\nu$, while for a fully chaotic system $\cP(\nu)$ should
fall off as $\nu^{-1}$ (pink noise). In fact, these power laws have been theoretically justified for
integrable systems with Poissonian spacings and Gaussian random matrix ensembles,
respectively~\cite{FGMMRR04}. The conjecture has also been confirmed in subsequent publications
for quantum billiards and several random matrix ensembles~\cite{SB05,MCD07,Re08}. Again, spin
chains of HS type seem to be somewhat exceptional also in this respect, since the power spectrum
of both the Haldane--Shastry and the Polychronakos--Frahm chain behaves as $\nu^{-4}$ (black
noise) for small $\nu$~\cite{BFGR10}.

We have evaluated numerically the power spectrum of Inozemtsev's chain with $N=18$ spins $1/2$ and
$k^2=0.05, 0.1,\dots,0.95$ in subspaces with\footnote{Since the total momentum is defined up to an
  integer multiple of $2\pi$ its possible values can be taken as $j\pi/9$, with
  $j=0,\pm1,\dots,\pm8,9$. The subspaces with $P_\tot=\pm j\pi/9$ for $j=1,\dots,8$ (and fixed $S^z$,
  $\bS^2$) have the same energies, and can therefore be merged. The two remaining subspaces with
  $j=0,9$ have been discarded, since their dimension is smaller than that of the other ones
  roughly by a factor of $2$. Similarly, the subspaces with $S=0$ and $S=4,5,\dots,9$ have been
  dropped due to their smaller dimension.} $S^z=0$, total spin~$S=1,2,3$ and $|P_\tot|=j\pi/9$
($j=1,\dots,8$). For each value of $k^2$, we have averaged $\cP(\nu)$ over the $3\times 8=24$
subspectra under consideration, removing an appropriate number of energies when necessary to make
all of them of equal length. As in the study of the spacings distribution in the previous
subsection, we have also dropped the largest 25\% spacings in each of these subspaces prior to the
computation of the power spectrum. Our analysis shows that in all cases $\cP(\nu)$ is proportional
to $\nu^{-\al}$ (for $1\le \nu\le n/4$), where the exponent $\al$ varies very slightly with $k^2$
between $2.097$ (for $k^2=0.15$) to $2.145$ (for $k^2=0.05$). For instance, for $k^2=1/2$ the
exponent $\al$ is equal to $2.100$, with coefficient of determination $r^2=0.9870$. For comparison
purposes, we have performed a similar analysis for the $\mathrm{su}(2)$ Heisenberg chain with
$N=18$ spins, finding that $\al=1.916$ with coefficient of determination $r=0.9882$ (see
Fig.~\ref{fig.Pnu}). These facts clearly suggest that Inozemtsev's chain is closer to more
``standard'' integrable systems like the Heisenberg chain than to the Haldane--Shastry ($k^2=0$)
or Polychronakos--Frahm chains, even for values of $k^2$ close to $0$.

\begin{figure}[ht]
  \centering
\includegraphics[width=8cm]{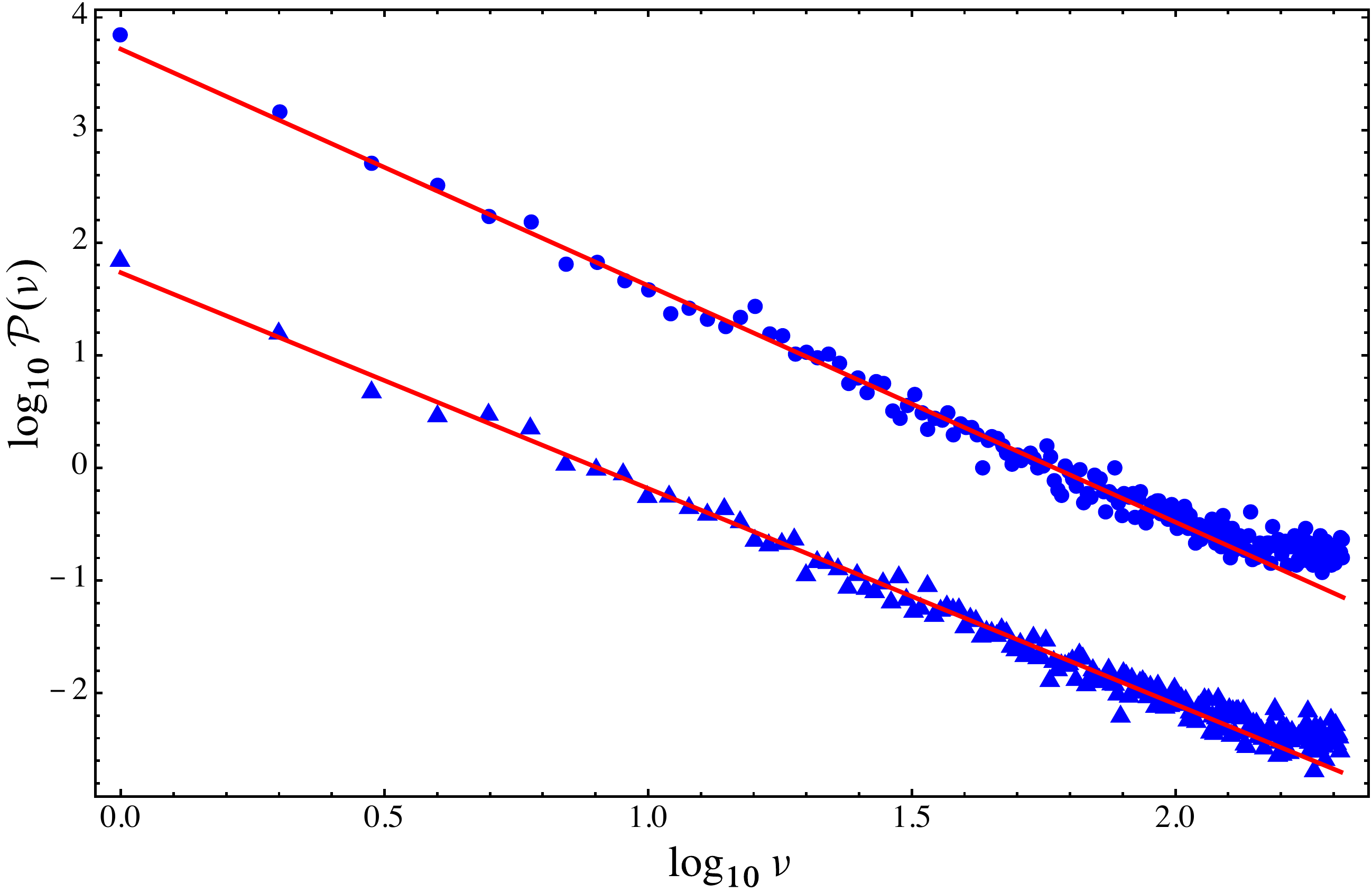}
\caption{Blue dots: $\log$-$\log$ plot of $\cP(\nu)$ for the spin chain~\eqref{cH} with $18$ spins
  $1/2$, $k^2=0.5$ and $S^z=0$, computed by averaging over 24 subspectra of $n=416$ spacings with
  quantum numbers $S=1,2,3$ and $|P_\tot|=\pi/9,2\pi/9,\dots,8\pi/9$. Blue triangles: analogous
  plot for the spin $1/2$ Heisenberg chain~\eqref{cHhe} with $18$ spins and $S^z=0$ (24 subspectra
  with $n=412$ spacings each, with the same quantum numbers as before). The solid red lines
  represent the best-fit straight lines to the data in the range $1\le \nu\le n/4$. (The plot for
  the Heisenberg chain has been lowered to avoid overlapping).}
\label{fig.Pnu}
\end{figure}%

\subsection{Average degeneracy}

It is well-known that the energy spectrum of all spin chains of HS type associated with the
$A_{N-1}$ root system is highly degenerate, much more so than is the case with more typical
integrable models like the Heisenberg chain. For the original Haldane--Shastry and
Polychronakos--Frahm chains, this high degeneracy is ultimately due to the invariance under
suitable realizations of the Yangian $Y\big(\mathrm{su}(m)\big)$ \cite{HHTBP92,Hi95npb}. An
important consequence of the Yangian symmetry of these chains is their equivalence to a classical
vertex model with energies given by~\cite{BBH10}
\begin{equation}\label{vM}
E_{\bn}=\sum_{i=1}^{N-1}\cE(i)\,\de(n_i,n_{i+1})\,,\qquad \bn\equiv(n_1,\dots,n_N)\in\{1,\dots,m\}^N\,,
\end{equation}
where
\begin{equation}\label{motifs}
\de(i,j)=\begin{cases}
    1,&i>j\\
    0, &i\le j
    \end{cases}
\end{equation}
and
\begin{equation}\label{cE}
\cE(i)=
\begin{cases}
  i(N-i)\,,& \HS\\
  i\,,& \PF\,.
\end{cases}
\end{equation}
In fact, Eq.~\eqref{vM} also holds for the Frahm--Inozemtsev (hyperbolic) chain~\cite{FI94} with
dispersion relation $\cE(i)=i(\al+i-1)$, where $\al$ is the positive parameter appearing in this
model~\cite{BFGR10,BBH10}. In view of these facts, it is natural to enquire whether the elliptic
chain~\eqref{cH} is equivalent to a vertex model~\eqref{vM}-\eqref{motifs} for a suitable choice
of the dispersion relation $\cE(i)$. As it turns out, the number of distinct levels
$\ell\equiv\ell(N,m)$ of the model~\eqref{vM}-\eqref{motifs} admits an upper bound which is
independent of the dispersion relation~$\cE(i)$. More precisely, it can be shown~\cite{FG14} that
\begin{equation}\label{upperbound}
  \ell\le F^{(m)}_{N+m-1},
\end{equation}
where $F^{(m)}_n$ is the $n$-th $m$-Fibonacci number~\cite{Mi60} defined by $F^{(m)}_0=\cdots=
F^{(m)}_{m-2}=0$, $F^{(m)}_{m-1}=1$, and
\[
F^{(m)}_n=\sum_{j=1}^m F^{(m)}_{n-j}\,,\quad n\ge m\,.
\]
Equivalently, the {\em average degeneracy} $\bar d\equiv m^N/\ell$
of the model~\eqref{vM}-\eqref{motifs} satisfies the inequality
\begin{equation}\label{bard}
\bar d\ge\frac{m^N}{F^{(m)}_{N+m-1}}\,,
\end{equation}
regardless of the functional form of the dispersion relation. In particular, a necessary condition
for the ${\mathrm su}(m)$ Inozemtsev chain to be equivalent to a vertex model of the
form~\eqref{vM}-\eqref{motifs} is that its average degeneracy $\bar d$ satisfy the
inequality~\eqref{bard}. In order to test this fact in the simplest case $m=2$, we have
numerically computed the average degeneracy of the ${\mathrm su}(2)$ chain~\eqref{cH} for
$10\le N\le18$ and $k^2=0.1,0.2,\dots,0.9$. In fact, for these values of $k^2$ we have found that
$\bar d$ is independent of $k^2$ (up to small numerical fluctuations), and thus can be regarded as
a function of $N$ only. Our calculations clearly show that the function $\bar d(N)$ is definitely
smaller than the RHS of Eq.~\eqref{bard} in the range under consideration;
cf.~Fig.~\ref{fig.degeneracies}. Thus for $10\le N\le 18$ the elliptic chain~\eqref{cH} cannot be
equivalent to a vertex model of the form~\eqref{vM}-\eqref{motifs}. Moreover, from
Fig.~\ref{fig.degeneracies} it is also apparent that the RHS of Eq.~\eqref{bard} grows much faster
with $N$ than the elliptic chain's average degeneracy $\bar d(N)$. This strongly suggests that
Inozemtsev's chain is not equivalent to \emph{any} vertex model~\eqref{vM}-\eqref{motifs} for
arbitrary $N$. This conclusion is consistent with the widespread belief that Inozemtsev's chain
does not possess the Yangian symmetry for finite values of $N$. As a further confirmation of this
assertion, we have numerically verified that for $10\le N\le18$ the average degeneracy of
Inozemtsev's $\mathrm{su}(2)$ chain essentially coincides with that of the spin~$1/2$ Heisenberg
chain~\eqref{cHhe}, which is known to be invariant under the Yangian only in an infinite
lattice~\cite{Be93} (see again~Fig.~\ref{fig.degeneracies}).
\begin{figure}[ht]
  \centering
\includegraphics[width=8cm]{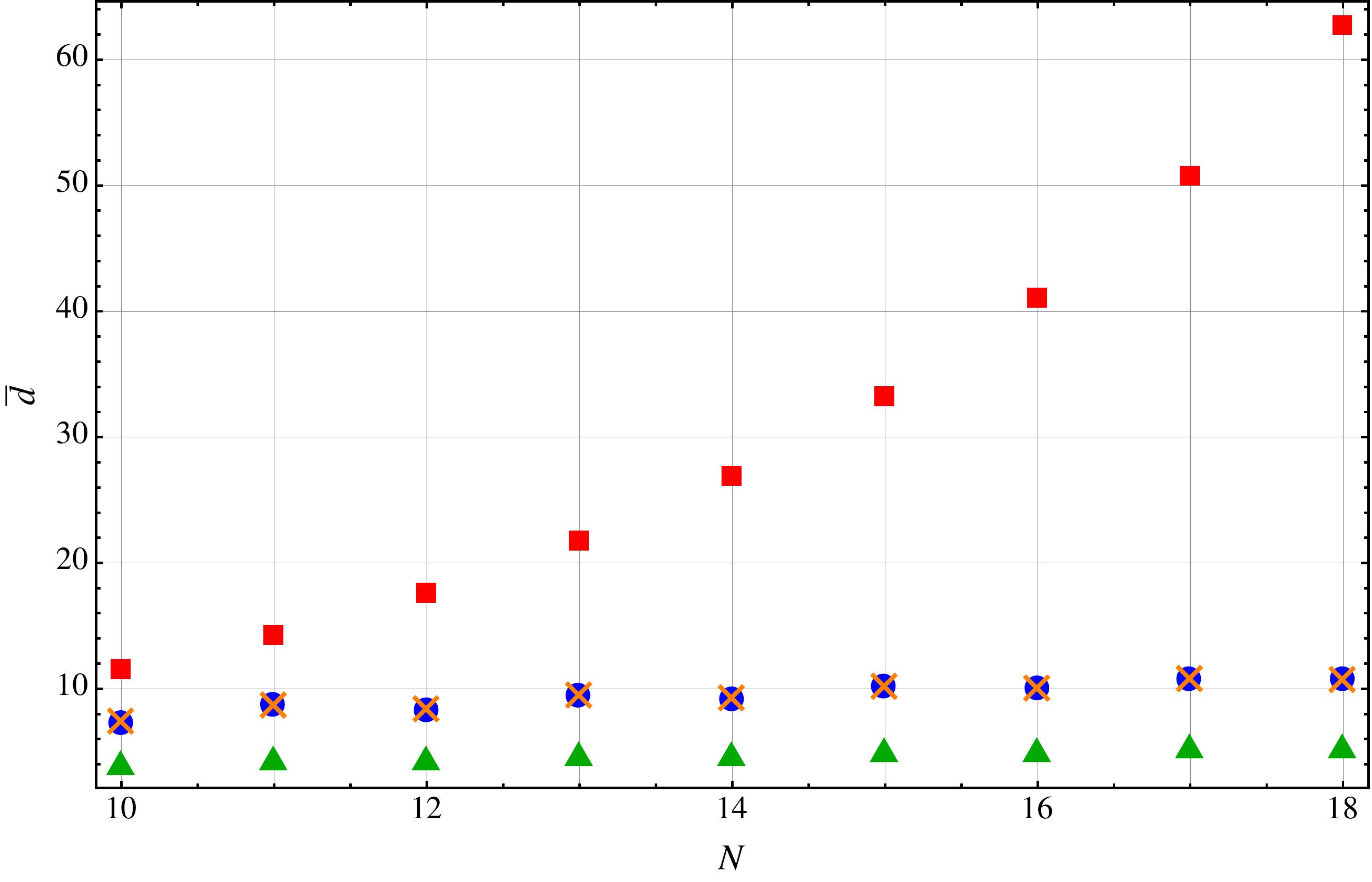}
\caption{Average degeneracy $\bar d$ of Inozemtsev's $\mathrm{su}(2)$ chain (blue dots) as a
  function of the number of spins $N$, vs.~the minimum average degeneracy of a vertex model of the
  form~\eqref{vM}-\eqref{motifs}, given by the RHS of Eq.~\eqref{bard} (red squares). For
  comparison purposes, we have also plotted the average degeneracy of the $\mathrm{su}(2)$
  Heisenberg chain (orange crosses) and of a rotationally and translationally
  invariant spin chain~\eqref{hchain} with random interactions (green triangles).}
\label{fig.degeneracies}
\end{figure}%

\section{Conclusions}\label{sec.conc}

In this paper we have analyzed the integrability of the spin chain with long-range elliptic
interactions introduced by Inozemtsev~\cite{In90}. Our approach bypasses the considerable
technical difficulties involved in the explicit construction of a complete set of commuting
first-integrals, relying instead on the analysis of several statistical properties of the spectrum
which have been conjectured to characterize integrable vs.~chaotic behavior. More precisely, we
have studied the distribution of the spacing between consecutive (unfolded) levels and the power
spectrum of the spectral fluctuations. In both cases our results are clearly consistent with the
generally held assumption that Inozemtsev's chain is indeed integrable. As mentioned in the
Introduction, the integrability of Inozemtsev's chain has interesting implications in the context
of the AdS-CFT correspondence, where this model has often been proposed as the integrable chain
describing non-perturbatively planar $\cN=4$ gauge theory.

Since Inozemtsev's spin chain tends to the Haldane--Shastry chain when the modulus $k$ tends to
zero, while it is related in a simple way to the Heisenberg chain when $k$ tends to one, it is
natural to compare it with its two limiting cases. Interestingly, our analysis shows that
Inozemtsev's chain shares certain properties with both the Heisenberg and the HS chains. Thus,
we have found that the distribution of the spacing between consecutive levels and the spectral
noise of Inozemtsev's chain clearly resemble those of the Heisenberg rather than the HS chain.
This similarity between the Inozemtsev and the Heisenberg chains is even more apparent if one
considers the average spectral degeneracy, as we have seen in the previous section. This is not
surprising, since the extremely high average degeneracy of the HS chain is related to its
underlying Yangian symmetry, which is absent in the Heisenberg and (most likely) the Inozemtsev
chains. In contrast, our numerical calculations strongly suggest that the level density of the
Inozemtsev chain becomes Gaussian as the number of spins tends to infinity. As we have mentioned
in Section~\ref{sec.LD}, this statistical property of the spectrum is shared by all spins chains
of HS type related to the $A_{N-1}$ root system.

Apart from the results just mentioned, most of which are ultimately of a numerical nature, we have
also derived several exact formulas that are also of interest from a mathematical viewpoint. Thus,
the exact calculation of the mean and the standard deviation of the spectrum in
Section~\ref{sec.LD} relies on the evaluation of the sum of certain powers of the Weierstrass
elliptic function. We have shown in~\ref{app.sums} how to compute in closed form these sums for
any positive power, and have also derived their asymptotic behavior when the number of terms tends
to infinity. Both of these results generalize their classical analogs for trigonometric
functions~\cite{BY02}, well known in the mathematical literature.

\section*{Acknowledgments}
This work was supported in part by Spain's MINECO under grant no.~FIS2011-22566. The authors would
like to thank Vladimir Inozemtsev for useful discussions on the elliptic chain's integrability and
Armando Rela\~no for enlightening conversations on quantum chaos.

\appendix
\section{Some useful identities relating Weierstrass and Jacobi elliptic
  functions}\label{app.prel}
Let $\wp(z;\om_1,\om_3)$ denote the Weierstrass's elliptic function~\cite{La89} with fundamental
half-periods $\om_1$ and $\om_3$ (where $\Im(\om_3/\om_1)>0$), defined by
\[
\wp(z;\om_1,\om_3)=\frac1{z^2}+\sum_{(l,n)\in\ZZ^2\setminus{\{0\}}}\bigg[\frac1{(z-2l\om_1-2n\om_3)^2}
  -\frac1{(2l\om_1+2n\om_3)^2}\bigg]\,.
\]
We shall adopt the abbreviated notation $\wp(z)$ when the periods of $\wp$ are clear from the
context. The function $\wp$ is even, doubly periodic and meromorphic, with a double pole on the
sites of the period lattice $2l\om_1+2n\om_3$ ($l,n\in\ZZ$). From the previous formula it is
obvious that $\wp$ is real on the real axis when $\om_1$ is real and $\om_3$ is pure imaginary,
and that it satisfies the homogeneity relation
\begin{equation}\label{wpdil}
\wp(\la z;\la\om_1,\la\om_3)=\la^{-2}\wp(z;\om_1,\om_3)\,.
\end{equation}
The function $\wp$ is minus the derivative of the Weierstrass $\ze$ function defined by
\[
\ze(z)\equiv\ze(z;\om_1,\om_3)=\frac1{z}+\sum_{(l,n)\in\ZZ^2\setminus{\{0\}}}\bigg[\frac1{z-2l\om_1-2n\om_3}
+\frac{1}{2l\om_1+2n\om_3}+\frac{z}{(2l\om_1+2n\om_3)^2}\bigg]\,,
\]
which is an odd meromorphic function with simple poles on the lattice $2l\om_1+2n\om_3$
($l,n\in\ZZ$). Note, however, that $\ze$ is \emph{not} $(2\om_1,2\om_3)$-periodic, but rather
verifies
\begin{equation}\label{quasi-p}
\ze(z+2\om_i)=\ze(z)+2\eta_i\,,
\end{equation}
where
\begin{equation}\label{etai}
\eta_i\equiv\eta_i(\om_1,\om_3)=\ze(\om_i;\om_1,\om_3)\,.
\end{equation}
The function $\ze$ obviously satisfies the homogeneity relation
\[
\ze(\la z;\la\om_1,\la\om_3)=\la^{-1}\ze(z;\om_1,\om_3)\,;
\]
in particular,
\begin{equation}
  \eta_i(\la\om_1,\la\om_3)=\la^{-1}\eta_i(\om_1,\om_3)\,.
  \label{hometai}
\end{equation}
In addition, the numbers $\eta_i$ are related by Legendre's identity
\begin{equation}
  \label{Legendre}
  \eta_1\om_3-\eta_3\om_1=\frac{\iu\ms\pi}2\,.
\end{equation}
The Laurent series of $\wp$ around the origin has the form
\begin{equation}
  \label{Laurentwp}
  \wp(z)=\frac1{z^2}+\frac{g_2}{20}\,z^2+\frac{g_3}{28}\,z^4+[z^6]\,,
\end{equation}
where $[z^6]$ denotes a function analytic at the origin with a zero of order at least $6$ at this
point, and the \emph{invariants} $g_i\equiv g_i(\om_1,\om_3)$ are given by
\[
g_2=60\sum_{(l,n)\in\ZZ^2\setminus\{0\}}(2l\om_1+2n\om_3)^{-4}\,,\qquad
g_3=140\sum_{(l,n)\in\ZZ^2\setminus\{0\}}(2l\om_1+2n\om_3)^{-6}\,.
\]
These definitions obviously imply the homogeneity relations
\[
g_2(\la\om_1,\la\om_3)=\la^{-4}g_2(\om_1,\om_3)\,,\qquad
g_3(\la\om_1,\la\om_3)=\la^{-6}g_2(\om_1,\om_3)\,.
\]
As is well known, the derivative of $\wp$ is an algebraic function of $\wp$, namely
\[
\wp'^2=4\wp^3-g_2\wp-g_3\,.
\]
From this identity it is straightforward to show that the invariants of $\wp$ are related to the
numbers
\begin{equation}\label{eis}
e_i=\wp(\om_i)\,,\qquad i=1,2,3\,,
\end{equation}
(where $\om_2\equiv -\om_1-\om_3$ by definition) by the well-known formula
\[
4s^3-g_2s-g_3=4(s-e_1)(s-e_2)(s-e_3)\,.
\]
It is important to recall that when $\om_1,\iu\om_3\in\RR$ all the numbers $e_i$, $g_i$ are real
and satisfy the inequalities
\begin{equation}
  e_1>e_2>e_3\,,\qquad g_3^2>27 g_2^3\,.
  \label{eidelta}
\end{equation}

We shall denote by $\sn(z,k)$ (or, in abbreviated fashion, $\sn z$) Jacobi's elliptic sine with
\emph{modulus} $k$. Its fundamental periods are $4K(k)$ and $2\iu K'(k)$, where
\begin{equation}\label{K}
K(k)=\int_0^{\pi/2}\frac{\diff x}{\sqrt{1-k^2\sin^2x}}
\end{equation}
is the complete elliptic integral of the first kind, and
\[
K'(k)=K(k')\,,\qquad k'\equiv\sqrt{1-k^2}\,.
\]
When $0\le k\le 1$ the numbers $K\equiv K(k)$ and $K'\equiv K'(k)$ are both real (with
$K(1)=K'(0)=\infty$), and $\sn z$ is real for real values of $z$. It is also well known that for
$k=0,1$ Jacobi's elliptic sine reduces to an elementary function, namely
\begin{equation}\label{snsin}
\sn(z,0)=\sin x\,,\qquad \sn(z,1)=\tanh z\,.
\end{equation}

The elliptic sine is related to the Weierstrass $\wp$ function by the well-known identity
\begin{equation}\label{wpsn}
\frac1{\sn^2 z}=\wp(z;K,\iu K')+\frac13(1+k^2)
\end{equation}
(cf.~\cite[Eq.~6.9.11]{La89}). Conversely, given two nonzero complex numbers $\om_1,\om_3$ with
$\Im(\om_3/\om_1)>0$ it is shown in the latter reference that
\begin{equation}\label{wpsngen}
\wp(z;\om_1,\om_3)=e_3+\frac{e_1-e_3}{\sn^2(\sqrt{e_1-e_3}\,z,k)}\,,
\end{equation}
where the modulus of the elliptic sine is determined by
\begin{equation}\label{kei}
k^2=\frac{e_2-e_3}{e_1-e_3}\,.
\end{equation}
In particular, when $\om_1$ and $\iu\om_3$ are both real the previous identity and the
inequalities~\eqref{eidelta} imply that $k^2$ is a real number in the interval $(0,1)$. Note also
that from Eq.~\eqref{wpsngen} it immediately follows the important relation
\[
\frac{\om_3}{\om_1}=\frac{\iu K'}{K}\equiv\iu\tau\,.
\]

\section{Extrema of translation-invariant functions}\label{app.lemmas}

\begin{lem}
  Let $U$ be a scalar real-valued function of class $C^2$ in an open subset $\Om\subset\RR^N$.
  Suppose that there is a fixed vector $\bv\in\RR^N$ such that $\Om$ is invariant under translations
  in the direction of $\bv$, and
  \[
  U(\bx+\la \bv)=U(\bx)\,,\qquad\all \bx\in\Om\,,\qquad\all\la\in\RR\,.
  \]
  If $\bx_0\in\Om$ is a critical point of $U$ satisfying
  \begin{equation}\label{D2gx0}
    \big(\bh,D^2U(\bx_0)\cdot \bh\big)>0\,,\qquad \all \bh\in(\RR \bv)^\perp\,,\en \bh\ne0\,,
  \end{equation}
  then $U$ has a local minimum at $\bx_0$.
\end{lem}
\begin{proof}
  Simply change variables so that $\bv$ is in the direction of a coordinate vector.
\end{proof}
Note that, since
\[
D^2U(\bx)\cdot \bv=0\,,
\]
we can replace~\eqref{D2gx0} by the apparently stronger condition
\[
\big(\bh,D^2U(\bx_0)\cdot \bh\big)>0\,,\qquad \all \bh\notin\RR \bv\,.
\]
\begin{lem}\label{lem.minU}
  Let $U$ and $\Om$ be as in the previous lemma. Suppose, moreover, that $\Om$ is convex and
  \[
  \big(\bh,D^2U(\bx)\cdot \bh\big)>0\,,\qquad \all \bx\in\Om\,,\quad\all \bh\in(\RR \bv)^\perp\,,\en
  \bh\ne0\,.
  \]
  Then $U$ has at most one critical point in $\Om$ modulo translations along the vector $\bv$, and
  this critical point (if it exists at all) is necessarily a global minimum.
\end{lem}
\begin{proof}
  Note, first of all, that the condition on the Hessian of $U$ is equivalent to
  \begin{equation}
    \label{D2gx}
    \big(\bh,D^2U(\bx)\cdot \bh\big)>0\,,\qquad \all \bx\in\Om\,,\quad\all \bh\notin\RR \bv\,.
  \end{equation}
  Suppose, to begin with, that $\bx_0,\by_0$ are critical points of $U$ in $\Om$ such that $\by_0-\bx_0$
  is not proportional to $\bv$. Since $\Om$ is convex, the segment $t\by_0+(1-t)\bx_0$ ($0\le t\le 1$)
  lies in $\Om$, so that the function
  \[
  \vp(t)=U\big(t\by_0+(1-t)\bx_0\big)
  \]
  has two critical points at $t=0,1$. By Rolle's theorem, there is a point $t_0\in(0,1)$ such that
  \[
  0=\vp''(t_0)=\big(\by_0-\bx_0,D^2U(t_0\by_0+(1-t_0)\bx_0)\cdot(\by_0-\bx_0)\big)\,.
  \]
  This contradicts condition~\eqref{D2gx}, since $\by_0-\bx_0$ is not proportional to $\bv$. Thus $U$
  has at most one critical point in $\Om$, modulo translations along $\bv$.

  Let now $\bx_0\in\Om$ be a critical point of $U$, which must be a local minimum on account of
  the previous lemma. To prove that $\bx_0$ is actually a global minimum, suppose that
  $U(\by_0)\le U(\bx_0)$ for some $\by_0\in\Om$ such that $\by_0-\bx_0$ is not proportional to
  $\bv$. The function $\vp(t)$ defined above now satisfies
  \[
  \vp(0)\ge\vp(1)\,,\qquad \vp'(0)=0\,,\qquad
  \vp''(0)=\big(\by_0-\bx_0,D^2U(\bx_0)\cdot(\by_0-\bx_0)\big)>0\,.
  \]
  Since $\vp$ is not constant on $[0,1]$ (indeed, $\vp''(0)>0$), the equality
  \[
  \vp(1)-\vp(0)=\int_0^1\vp'(s)\diff s
  \]
  and the continuity of $\vp'$ implies that there is a point $t_0\in(0,1)$ such that
  $\vp'(t_0)<0$. By the mean value theorem, this in turn implies that there is a second point
  $t_1\in(0,t_0)$ such that
  \[
  0>\vp''(t_1)=\big(\by_0-\bx_0,D^2U(t_1\by_0+(1-t_1)\bx_0)\cdot(\by_0-\bx_0)\big)\,.
  \]
  This again contradicts the hypothesis, since $\by_0-\bx_0\notin\RR\bv$.
\end{proof}

\section{Evaluation of two sums involving the Weierstrass elliptic function}
\label{app.sums}

In this appendix we shall evaluate the two sums
\begin{equation}
  \label{Spp}
  S_p'=\sum_{j=1}^{N-1}\wp_N^p(j)\,,\qquad p=1,2\,,
\end{equation}
which are used in Section~\ref{sec.LD} to compute the mean and the standard deviation of the
spectrum of the Inozemtsev spin chain~\eqref{cH}. In the latter equation $\wp_N$ denotes the
Weierstrass function with periods $N$ and $2\om_3$, where $\om_3\in\iu\,\RR_+$.

For the $p=1$ case, let $\ga$ denote the (positively oriented) perimeter of the rectangle with
vertices $\pm\om_3$ and $N\pm\om_3$, with a semicircular indentation of radius less than
$\min(1,|\om_3|)$ to the left of the points $0$ and $N$ (see Fig.~\ref{fig.contour-new}).
\begin{figure}[b]
  \centering
  \includegraphics{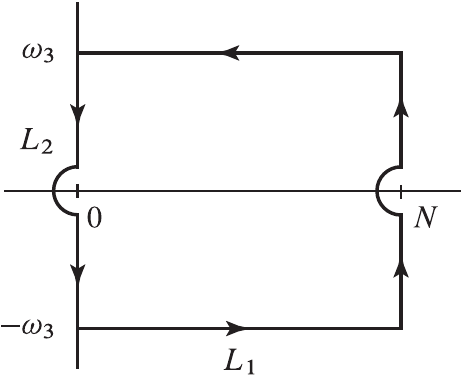}
  \caption{Contour of integration used in the computation of the sum $S_1'$.}
  \label{fig.contour-new}
\end{figure}%
The
function
\[
f(z)=\wp_N(z)\ze_1(z)\,,
\]
where $\ze_1$ denotes the Weierstrass zeta function with periods $1$ and $2\om_3$, is meromorphic
in the complex plane with poles on the lattice $j+2l\om_3$, $j,l\in\ZZ$. The residues of $f$ at
these points can be easily computed by taking into account Eq.~\eqref{Laurentwp} and the
well-known Laurent series
\begin{equation}\label{Laurentzeta}
  \ze_1(z)=\frac1z+[z^3]\,.
\end{equation}
Indeed
\[
\wp_N(z)\ze_1(z)=\frac1{z^3}+[z]\en\implies\en \res(f,0)=0\,,
\]
while at $z=j\in\{1,\dots,N-1\}$ the analyticity of $\wp_N$ and the quasi-periodicity of $\ze_1$
imply
\[
\wp_N(z)\ze_1(z)=\wp_N(z)\big(\ze_1(z-j)+2j\,\ze_1(1/2)\big)\en\implies\en \res(f,j)=\wp_N(j)\,,
\qquad j=1,\dots,N-1\,.
\]
Applying Cauchy's residue theorem we obtain
\begin{equation}\label{CRT}
\frac1{2\pi\iu}\int_\ga f(z)\,\diff z=\sum_{j=1}^{N-1}\wp_N(j)\equiv S_1'\,.
\end{equation}
The integral in the LHS can be readily computed taking into account the (quasi-)periodicity
properties of the Weierstrass functions. Indeed, denoting by $L_1$ and $L_2$ the bottom and left
sides of the contour $\ga$, by Eqs.~\eqref{quasi-p} we have
\begin{align*}
\int_\ga f(z)\,\diff z &= \int_{L_1}\wp_N(z)\,\big(\ze_1(z)-\ze_1(z+2\om_3)\big)\,\diff z
+\int_{L_2}\wp_N(z)\,\big(\ze_1(z)-\ze_1(z+N)\big)\,\diff z\\
&=-2\ze_1(\om_3)\int_{L_1}\wp_N(z)\,\diff z-2N\ze_1(1/2)\int_{L_2}\wp_N(z)\,\diff z\\
&=2\ze_1(\om_3)\,\ze_N(z)\,\Big|_{-\om_3}^{N-\om_3}+2N\ze_1(1/2)\,\ze_N(z)\,\Big|_{\om_3}^{-\om_3}
=4\big[\ze_1(\om_3)\ze_N(N/2)-N\ze_1(1/2)\ze_N(\om_3)\big]\,.
\end{align*}
{}From Legendre's relation~\eqref{Legendre} applied to both $\ze_1$ and $\ze_N$ we
have
\begin{align*}
4\om_3\big[\ze_1(\om_3)\ze_N(N/2)-N\ze_1(1/2)\ze_N(\om_3)\big]
&= 4\ze_1(\om_3)\bigg(\frac N2\,\ze_N(\om_3)+\frac{\iu\pi}2\bigg)
-4N\ze_N(\om_3)\bigg(\frac 12\,\ze_1(\om_3)+\frac{\iu\pi}2\bigg)\\
&=2\pi\iu\,\big(\ze_1(\om_3)-N\ze_N(\om_3)\big)\,.
\end{align*}
Using this identity in Eq.~\eqref{CRT} we finally obtain the formula
\[
S_1' = \frac1{\om_3}\big(\ze_1(\om_3)-N\ze_N(\om_3)\big)\,.
\]
With the help of Legendre's identity, we can rewrite this equation in the equivalent form
\begin{equation}
  \label{S1p}
  S_1' = 2\ms\big(\ze_1(1/2)-\ze_N(N/2)\big)\,.
\end{equation}

In order to evaluate the sum~\eqref{Spp} for $p=2$, consider to begin with the function
\[
g(z)=\sum_{j=0}^{N-1}\wp_N(z+j)\,.
\]
This function is easily seen to have periods $1$ and $2\om_3$, and is analytic everywhere except
for double poles at the points $j+2l\om_3$ ($j,l\in\ZZ$) on its fundamental period lattice.
Moreover, by Eq.~\eqref{Laurentwp} the principal part of $g$ at these poles is $(z-j)^{-2}$. Since
the Weierstrass function $\wp_1(z)$ with periods $1$ and $2\om_3$ has the same poles and principal
parts as $g$, by Liouville's theorem we have
\[
\sum_{j=0}^{N-1}\wp_N(z+j)=\wp_1(z)+c
\]
(where the constant $c$ is in fact the sum $S_1'$). Differentiating this equality twice and taking
into account the identity
\[
\wp''=6\wp^2-\frac12\,g_2
\]
we immediately obtain
\begin{equation}\label{wp2z+j}
\sum_{j=1}^{N-1}\wp_N^2(z+j)=\wp_1^2(z)-\wp_N^2(z)+\frac{N}{12}\,g_2(N/2,\om_3)-\frac1{12}\,g_2(1/2,\om_3)\,.
\end{equation}
The limit as $z\to0$ of the LHS is the sum $S_2'$. As to the RHS, using the Laurent
series~\eqref{Laurentwp} we readily find
\[
\lim_{z\to0}\big[\wp_1^2(z)-\wp_N^2(z)\big]=\frac1{10}\,\big[g_2(1/2,\om_3)-g_2(N/2,\om_3)\big]\,.
\]
{}From Eq.~\eqref{wp2z+j} we thus obtain
\begin{equation}\label{wp2sum}
  \sum_{j=1}^{N-1}\wp_N^2(j)=\frac1{12}\,\bigg(N-\frac65\bigg)\,g_2(N/2,\om_3)+\frac1{60}\,g_2(1/2,\om_3)\,.
\end{equation}

Proceeding in the same fashion, it is straightforward (albeit lengthy, unless $p$ is small) to
compute the sum $S_p'$ for any fixed value of $p$. For arbitrary $p$, it is not hard to show that
\begin{equation}
  \label{Sppasymp}
  S_p'=\frac{(2\pi)^{2p}|B_{2p}|}{(2p)!}+O(N^{-1})\,,
\end{equation}
where $B_n$ denotes the $n$-th Bernoulli number. The proof is based on the absolute convergence in
the punctured disk $0<|z|<\min(1,\tau)$ of the Laurent series~\cite{SS03}
\begin{equation}\label{wpEis}
  \wp(z;1/2,\iu\tau/2)\equiv\wp(z)=\frac1{z^2}+\sum_{l=1}^\infty(2l+1)E_{2l+2}(\tau)\,z^{2l}\,,
\end{equation}
where $E_{j}$ is the Eisenstein series
\[
E_{j}(\tau)=\sum_{(l,n)\in\ZZ^2\setminus\{0\}}(l+\iu\ms n\tau)^{-j}\,,\qquad j>2\,.
\]
Indeed, suppose first that $k^2\le1/2$, so that $\tau\equiv K'/K\ge1$. We then have
\[
\wp_N(j)\equiv\wp(j;N/2,\iu N\tau/2)=N^{-2}\wp(j/N)=
\frac1{j^2}+\frac1{N^2}\sum_{l=1}^\infty(2l+1)E_{2l+2}(\tau)\,(j/N)^{2l}\,, \qquad
1\le j\le N-1\,,
\]
since in this case the series in the RHS of Eq.~\eqref{wpEis} converges inside the unit disk.
Moreover, from the absolute convergence of the latter series in its disk of convergence it follows
that
\[
\bigg|\wp_N(j)-\frac1{j^2}\bigg|\le\frac1{N^2}\sum_{l=1}^\infty(2l+1)\big|E_{2l+2}(\tau)\big|\,(1/2)^{2l}
=O(N^{-2})\,,\qquad 1\le j\le [N/2]\,,
\]
and therefore
\[
\wp_N^p(j)=\frac1{j^{2p}}+O(N^{-2})\,,\qquad 1\le j\le [N/2]\,.
\]
Summing over $j$ and taking into account that $\wp_N(j)=\wp_N(N-j)$ we easily
obtain
\begin{align*}
S_p'&=2\sum_{j=1}^{[N/2]}\wp_N^p(j)+(\pi(N)-1)\,\wp_N^p(N/2)=
2\sum_{j=1}^{[N/2]}\frac1{j^{2p}}+O(N^{-1})+\big(\pi(N)-1\big)\,N^{-2p}\wp^p(1/2)\\
&=2\ze(2p)+O(N^{-1})\,,
\end{align*}
where $\ze$ denotes Riemann's zeta function\footnote{Recall that if $\la>0$ we have
  $\ds\sum_{j=\la N+1}^\infty j^{-2p}<\int_{\la N}^\infty x^{-2p}\,\diff x=O(N^{1-2p})$.}. This is
essentially Eq.~\eqref{Sppasymp}.

Suppose, on the other hand, that $k^2>1/2$, so that $\tau=K'/K<1$. In this case we can write
\[
S_p'=2\sum_{j=1}^{[N\tau/2]}\wp_N^p(j)+(\pi(N)-1)N^{-2p}\wp^p(1/2)+R_p\,,
\]
with
\[
R_p=2\sum_{j=[N\tau/2]+1}^{[N/2]}\wp_N^p(j)\le\big(N(1-\tau)+2\big)\,\wp_N^p(N\tau/2)
=\big(N(1-\tau)+2\big)\,N^{-2p}\wp^p(\tau/2)=O(N^{1-2p})\,.
\]
On the other hand, for $j=1,\dots,[N\tau/2]$ we again have
\[
\bigg|\wp_N(j)-\frac1{j^2}\bigg|\le\frac1{N^2}\sum_{l=1}^\infty(2l+1)\big|E_{2l+2}(\tau)\big|\,(\tau/2)^{2l}=O(N^{-2})\,,
\]
on account of the absolute convergence of the power series in the RHS of Eq.~\eqref{wpEis} inside
its disk of convergence $|z|<\tau$. Thus in this case
\[
\wp_N^p(j)=\frac1{j^{2p}}+O(N^{-2})\,,\qquad 1\le j\le [N\tau/2]\,,
\]
from which Eq.~\eqref{Sppasymp} easily follows as before.

\section{Simplification of several sums appearing in Eqs.~\eqref{ga1} and~\eqref{ga2}}

In this appendix we shall simplify several sums appearing in Eqs.~\eqref{ga1}-\eqref{ga2} for the
skewness and kurtosis of the spectrum of a spin chain of the form~\eqref{hchain}. Consider, to
begin, with, the last sum in Eq.~\eqref{ga1}, which can be written as
\[
\sump{i,j,l} h_{ij}h_{jl}h_{li}=6\sum_{i<j<l}h(j-i)h(l-j)h(l-i)
=6\sum_{2\le a+b\le N-1}(N-a-b)\,h(a)\ms h(b)\ms h(a+b)\,.
\]
Performing the change of index $a\mapsto c=N-a-b$ in the latter sum and making use of the
identities~\eqref{hids} satisfied by the function $h$ we obtain
\[
\sum_{2\le a+b\le N-1}(N-a-b)\,h(a)\ms h(b)\ms h(a+b)
=\sum_{2\le b+c\le N-1}c\,h(b)\ms h(c)\ms h(b+c)\,,
\]
and therefore
\[
\sum_{2\le a+b\le N-1}a\,h(a)\ms h(b)\ms h(a+b)=\frac N3\,\sum_{2\le a+b\le N-1}h(a)\ms h(b)\ms
h(a+b)\,.
\]
We thus have
\begin{equation}\label{sumskew}
\sump{i,j,l} h_{ij}h_{jl}h_{li}=2N\sum_{2\le a+b\le N-1}h(a)\ms
h(b)\ms h(a+b)=2N\sum_{1\le i<j\le N-1}h(i)\ms
h(j)\ms h(j-i)\,.
\end{equation}
Consider next the second sum in Eq.~\eqref{ga2}:
\begin{align*}
  \sump{i,j,l} h_{ij}^2h_{jl}^2&=2\sum_{i<j<l}\big[h^2(j-i)\ms h^2(l-j)
  +h^2(j-i)\ms h^2(l-i)+h^2(l-j)\ms h^2(l-i)\big]\\
  &=2\sum_{2\le a+b\le N-1}(N-a-b)\,\big[h^2(a)\ms h^2(b)+h^2(a)\ms h^2(a+b)+h^2(b)\ms h^2(a+b)\big]\\
  &=2\sum_{2\le a+b\le N-1}a\,\big[h^2(a)\ms h^2(b)+h^2(a)\ms h^2(a+b)+h^2(b)\ms h^2(a+b)\big]\,.
\end{align*}
Hence
\begin{align*}
\sum_{2\le a+b\le N-1}(N-a-b)\,&\big[h^2(a)\ms h^2(b)+h^2(a)\ms h^2(a+b)+h^2(b)\ms h^2(a+b)\big]\\&=
\frac N3\,\sum_{2\le a+b\le N-1}\big[h^2(a)\ms h^2(b)+2h^2(a)\ms h^2(a+b)\big]\\
&=\frac N6\,\sum_{\substack{1\le a,b\le N-1\\ a+b\ne N}}h^2(a)\ms
h^2(b)+\frac N3\,\sum_{\substack{1\le a,b\le N-1\\ a\ne b}}h^2(a)\ms h^2(b)=\frac N2\,(S_2^2-S_4)
\,,
\end{align*}
and therefore
\begin{equation}
  \sump{i,j,l} h_{ij}^2h_{jl}^2=N(S_2^2-S_4)\,.
  \label{sum2}
\end{equation}
The third sum in Eq.~\eqref{ga2} can be dealt with similarly:
\begin{align}
  \sump{i,j,l} h_{ij}^2h_{jl}h_{li}&=2\sum_{i<j<l}\big[
  h^2(j-i)\ms h(l-j)\ms h(l-i)+h^2(l-i)\ms h(j-i)\ms h(l-j)+h^2(l-j)\ms h(j-i)\ms h(l-i)\big]\nonumber\\
  &=2\sum_{1\le a+b\le N-1}(N-a-b)\big[h^2(a)\ms h(b)\ms h(a+b)+h^2(a+b)\ms h(a)\ms h(b)+h^2(b)\ms h(a)\ms h(a+b)\big]\nonumber\\
  &=\frac{2N}3\,\sum_{1\le a+b\le N-1}\big[h^2(a+b)\ms h(a)\ms h(b)+2h^2(a)\ms h(b)\ms
  h(a+b)\big]\nonumber\\
  &=\frac{2N}3\,\sum_{1\le i<j\le N-1}\big[h^2(j)\ms h(i)\ms h(j-i)+2h^2(i)\ms h(j)\ms
  h(j-i)\big]\nonumber\\
  &=2N\sum_{1\le i<j\le N-1}h^2(i)\ms h(j)\ms h(j-i)
  =2N\sum_{1\le i<j\le N-1}h^2(j)\ms h(i)\ms h(j-i)\,,
  \label{sum3}
\end{align}
where in the last two steps we have made use of the identity $h(l)=h(N-l)$. Likewise, the last
(cyclic) sum in Eq.~\eqref{ga2} can be written as
\begin{align*}
\sump{i,j,l,n} h_{ij}h_{jl}h_{ln}h_{ni}&=8
\sum_{i<j<l<n}\big[h(j-i)\ms h(l-j)\ms h(n-l)\ms h(n-i)
\nonumber\\
&\hphantom{\sum_{i<j<l<n}\big[}
+h(j-i)\ms h(l-i)\ms h(n-j)\ms h(n-l)+h(l-i)\ms h(l-j)\ms h(n-i)\ms h(n-j)\big]\,,
\end{align*}
where each of the sums in the RHS can be easily simplified. Indeed, for the first of these
sums we have
\begin{align*}
  \sum_{i<j<l<n}h(j-i)\ms h(l-j)\ms h(n-l)\ms h(n-i)&=\sum_{3\le a+b+c\le N-1}(N-a-b-c)\ms
  h(a)\ms h(b)\ms h(c)\ms h(a+b+c)\nonumber\\ &=\frac N4\sum_{3\le a+b+c\le N-1}
  h(a)\ms h(b)\ms h(c)\ms h(a+b+c)\,,
\end{align*}
and similarly for the remaining two sums:
\begin{align*}
  &\sum_{i<j<l<n}\big[h(j-i)\ms h(l-i)\ms h(n-j)\ms h(n-l)+h(l-i)\ms h(l-j)\ms h(n-i)\ms h(n-j)
  \big]\nonumber\\&\hspace*{1.5em}=\sum_{3\le a+b+c\le N-1}(N-a-b-c)\ms h(a)\ms
  h(a+c)\ms\big[h(b)\ms h(b+c)+ h(a+b)\ms h(a+b+c)\big]\nonumber\\ &\hspace*{1.5em}=\frac
  N4\sum_{3\le a+b+c\le N-1}\ms h(a)\ms h(a+c)\ms\big[h(b)\ms h(b+c)+
  h(a+b)\ms h(a+b+c)\big]\\
  &\hspace*{1.5em}=\frac{N}2\sum_{3\le a+b+c\le N-1}\ms h(a)\ms h(b)\ms h(a+c)\ms h(b+c)\,,
\end{align*}
where the last equality follows from the change of index $b\mapsto b'=N-a-b-c$ in the second
sum. We thus finally obtain
\begin{equation}
  \sump{i,j,l,n} h_{ij}h_{jl}h_{ln}h_{ni}=2N\sum_{3\le a+b+c\le N-1}h(a)\ms h(b)\big[h(c)\ms h(a+b+c)
  +2h(a+c)h(b+c)\big]\,.
  \label{sum45}
\end{equation}


\end{document}